\documentclass[sn-mathphys-num]{sn-jnl}
\usepackage{tikz}
\usepackage{amssymb}
\usepackage{amsmath}
\usepackage{amsthm}
\usepackage{mathtools}
\usepackage{array}
\usepackage{graphicx}
\usepackage{hyperref}

\usepackage{mathrsfs}%
\usepackage[title]{appendix}%
\usepackage{xcolor}%
\usepackage{textcomp}%
\usepackage{manyfoot}%
\usepackage{booktabs}%
\usepackage{listings}%

\newtheorem{theorem}{Theorem}
\newtheorem{conjecture}{Conjecture}
\newtheorem{corollary}[theorem]{Corollary}
\newtheorem{lemma}[theorem]{Lemma}

\newtheorem{definition}{Definition}

\DeclareMathOperator*{\argmax}{\text{argmax}}

\DeclareMathOperator{\uu}{\mathbf{u}}

\usetikzlibrary{calc, shapes, fit,arrows.meta}

\title[Insights into the Structured Coordination Game with Neutral Options through Simulation]{Insights into the Structured Coordination Game with Neutral Options through Simulation}

\author*[1,3]{\fnm{John S.} \sur{McAlister}}\email{jmcalis6@vols.utk.edu}

\author[1,2,3]{\fnm{Nina H.} \sur{Fefferman}}

\affil*[1]{\orgdiv{Department of Mathematics}, \orgname{University of Tennessee - Knoxville}}

\affil[2]{\orgdiv{Department of Ecology and Evolutionary Biology}, \orgname{University of Tennessee - Knoxville}}

\affil[3]{\orgname{National Institute for Modeling Biological Systems (NIMBioS)}}
\date{May 2024}

\abstract{Coordination games have been of interest to game theorists, economists, and ecologists for many years to study such problems as the emergence of local conventions and the evolution of cooperative behavior. Approaches for understanding the coordination game with discrete structure have been limited in scope, often relying on symmetric reduction of the state space, or other constraints which limit the power of the model to give insight into desired applications. In this paper, we introduce a new way of thinking about equilibria of the structured coordination game with neutral strategies by means of graph partitioning. We begin with a few elementary game theoretical results and then catalogue all the Nash equilibria of the coordination game with neutral options for graphs with seven or fewer vertices. We extend our observations through the use of simulation on larger Erd\H{o}s-R\'enyi random graphs to form the basis for proposing some conjectures about the general relationships among edge density, cluster number, and consensus stability.}


\begin{document}

\maketitle

\section{Introduction}\label{introduction}
    The Coordination game, at its most basic, is a two player game with the following payoff matrix (Fig. \ref{twobytwopayoff})
    under the assumption that $a>d$ and $b>c$.
    \begin{figure}[h!]\label{twobytwopayoff}
		\begin{center} 
			\begin{tabular}{c|cc}
				&A&B\\
				\hline 
				A&a,a&c,d\\
				B&d,c&b,b
			\end{tabular}
		\end{center}
		\caption{The Payoff matrix for the $2\times 2$ coordination game}
    \end{figure}
    It has two pure strategy Nash equilibria, $(A,A), (B,B)$, and a mixed strategy Nash equilibrium, where both players play strategy $A$ with probability $p=\frac{b-d}{a+b-c-d}$. On its own, the two player game with two strategies is entirely understood\cite{Maschler2013}. Economists like Kandori, Mishihari and Rob considered how the dynamics of this game change when extended beyond a dyadic interaction to instead consider randomly paired interactions among $n$ players \cite{Kandori1993,Kandori1995,Robson1995}. The results they present show that, in these cases, when all players interact uniformly with one another and there are only two strategies available to pick from, the game will always converge to a consensus equilibrium through myopic best response with high inertia and $\varepsilon$-noise (i.e. one player updates their strategy to the best response and, rarely, a player may spontaneously change their strategy). 

    To understand the game in a more realistic setting, Ellison added structure to the game by assuming that players interact with some individuals more than others. This was accomplished by considering the game on a graph with players as vertices and with edges between players who may interact\cite{Ellison1993}. With structure like a square lattice or the so called ``circular city," the state space can be reduced and similar results about convergence to a consensus equilibrium can be found through the Radius-Coradius Theorem \cite{Ellison2000}. In addition to adding structure to the model, Ellison also relaxed the high inertia assumption and had all individuals update their strategies simultaneously. In the years following, many similar results were shown in slightly modified settings \cite{Oechssler1997,Oechssler1999,Ely2002} and compared to some empirical studies \cite{Szkup2020}. Which, up until 2010, are well summarized in \cite {Weidenholzer}.

    It was the goal of each of the models above to understand the dynamics of this game when there are strategies which offer different payoffs. Kandori et al., Ellison, and many of their predecessors have been able to determine when that the system will converge to the Risk dominant Nash equilibrium under certain random player pairings schemes. Further papers seeking to understand this game have followed in that same tradition \cite{Arditti2024} \cite{Buskens2016} and have often ignored the interesting case in which there are many strategies available but none offers a better payoff than the other (i.e. for any pair of players the payoff matrix is the identity matrix). This critical case brings with it the possibility of more complicated equilibrium states and may help us answer question about the emergence of conventions which do not provide a direct payoff.  Kadnori et al. proved that in the unstructured case, when the payoff matrix is $I_2$ then the game will converge to one of the two consensus equilibria with equal probability. In a structured setting, one study considered a neutral strategy. in addition to the two strategies traditionally considered, and was able to discuss properties of phase transitions from transient dynamics to equilibrium \cite{Szabo2016}, but their treatment of the problem is constrained only to the square lattice.  

    Seeking to fill this gap we have examined the structured coordination game, in which every strategy is neutral, on general graphs,  both analytically and numerically. In Section \ref{Background}, we describe the game more fully and give some simple analytical results. In section \ref{Catalogue}, we present a catalogue of all of the Nash equilibria of the game, which we call equilibrium partitions, for groups no larger than 7. In section \ref{Simulation}, we use a larger simulation, sweeping over many graphs on greater numbers of vertices, to understand trends between features of the graphs like edge density, diameter, centralization etc. and features of the partitions they admit. Finally, in section \ref{Discussion}, we offer some discussion of the results through conjectures which are active areas of further investigation.  We highlight how this critical case in the structured coordination game fits into the context of the previous work. By establishing new ways to discuss this system, we get closer to a more complete theoretical understanding of coordinated behavior and thereby enhance relevance to application areas such as the evolution of cooperation, spatially distributions of language and convention, and local conventions among cooperating groups.

\section{Background Analytical Results}\label{Background}  For the structured coordination game with neutral options, we imagine each individual as a vertex in a graph with edges connecting to the other individuals with whom they interact. The payoff of a strategy is determined only by the number of neighbors a player has that are using that same strategy. In this way, the payoff function is very simple. 
For a connected graph $G(V,E)$ each vertex $v\in V$ plays a strategy $c$ from a set of pure strategies $C$, and the payoff for $v$ is given by 
 \begin{equation}\label{fitness1}
		w(v,c|\uu)=|\{x\in \Gamma (v);\uu_x=c\}|
\end{equation}
where $\uu$ is the strategy profile, $\uu_x$ is the strategy that vertex $x$ is using, and $\Gamma(x)$ is the set of vertices which share an edge with $x$.

Using a myopic best response update rule we describe an initial value problem whose (non-unique) solutions satisfy \begin{equation}\label{dsys}
			\uu_v(t+1)\in \argmax_{c\in C}\{w(v,c|\uu(t))\}.
		\end{equation}
		It may be that $|\argmax|>1$, so we break ties in the following way, which could be described as $\varepsilon$-inertia (i.e. there is a very small cost of changing strategies which can always be overcome if changing strategies results in a payoff increase, but cannot be overcome if changing strategies results in no change in payoff). 
        If $\uu_v(t)\in \argmax_{c\in C}\{w(v,c|\uu(t))\}$, then $ \uu_v(t+1)=\uu_v(t)$. If not, $\uu_x(t+1)$ is selected from $\argmax_{c\in C}\{w(v,c|\uu(t))\}$ uniform randomly. Notice that we cannot impose a tie breaking rule like those described in \cite{Nissan2011} exactly because of the assumption the each strategy is neutral.

  There are two trivial facts which are important to state now.
  \begin{lemma} Equilibria of the dynamical system \eqref{dsys} are Nash equilibria of the structured coordination game with neutral options.
  \end{lemma} 
  \begin{lemma} Every graph admits at least one equilibrium which is the strategy profile where every individual is using the same strategy. We will call this the consensus equilibrium.
  \end{lemma}
  
Both of these facts are immediate from the definitions of equilibrium and Nash equilibrium. 

There is a natural correspondence between strategy profiles and vertex partitions. For any strategy profile, the corresponding vertex partition is the partition wherein all the vertices using the same strategy are in the same part of the partition. We find that using the language of vertex partitions helps alleviate redundancy borne from the fact that interchanging strategies is only a relabeling and does not change the game state in any way.  We will note a partition as $Q:=\{q^c\}_{c\in C}$, where $q^c$ denotes the set of vertices using strategy $c$ (also called a cluster). In this way, we can discuss equivalence classes of strategy profiles under relabelling rather than particular strategy profiles. The focus of the present study is to understand those partitions which correspond to Nash equilibria.  

\begin{definition}[Equilibrium Partition]
    An equilibrium partition is a partition which corresponds to an equivalence class of strategy profiles which are Nash Equilibria. 
\end{definition}
\begin{definition}[Indecomposable]
    If the only equilibrium partition of a graph is the trivial partition, that graph is said to be Indecomposable. Naturally, if non-trivial equilibrium partitions exist, the graph is said to be decomposable. 
\end{definition}
Little is known, in general, about equilibrium partitions but there are two more lemmata which will help our later discussion. 
\begin{lemma}\label{clusters of size 1}
    If $Q=\{q^c\}_{c\in C}$ is an equilibrium partition of a connected graph of order greater than 1, then $|q^c|\neq 1$.
\end{lemma}

The proof of lemma \ref{clusters of size 1} is trivial.

\begin{lemma}\label{connected clusters}
    If $Q=\{q^c\}_{c\in C}$ is an equilibrium partition and the subgraph spanned by $q^i$ is disconnected, then, if $q^i_1,...,q^i_n$ are the connected components of $q^i$, $\tilde{Q}=\{q^c\}_{c\neq i}\cup \{q^i_j\}_{j=1}^n$ is also an equilibrium partition. 
\end{lemma}
Before we give the proof, we present some notation for simplicity. Let $\partial q^c$ be the the set of those vertices in $q^c$ which have neighbors which are not in $q^c$. Moreover, let $\overline{\partial q^c}$ be those vertices not in $q^c$ which have neighbors in $q^c$.
\begin{proof}
    Suppose that $P=\{p^i\}_{i=1}^m$ is an equilibrium partition corresponding to the strategy profile $\uu$, and say the subgraph spanned by $p^1$ is disconnected. We will show that refining the partition so that every part corresponds to a connected subgraph results in an equilibrium partition. It will be no loss of generality to suppose the original equilibrium partition has only 1 disconnected part, so $p^i$ corresponds to a connected subgraph for $i>1$. Refine the partition to $\Tilde{P}=\{p^1_1,...,p^1_n\}\cup\{p^i\}_{i=2}^m$ and consider the corresponding strategy profile $\tilde{\uu}$. The new parts, $p^1_1,...,p^1_n$, correspond to new strategies which were not previously present in the original strategy profile. Any vertex $v\in p^i\setminus \overline{\partial p^1}$ for $i>1$ will be unaffected by this refinement, so we need only consider those vertices in $p^1$ and $\overline{\partial p^1}$.

    First we will consider a vertex  $v\in p^c\cap\partial p^1$. Observe that $w(v,c|\tilde{\uu})=w(v,c|\uu)=\geq w(v,1|\uu)$. Trivially, $p^1_i\subset p^1$ for all $i$, so $|p^1_i\cap\Gamma(v)|\leq|p^1\cap\Gamma(v)|$ for all $i$. This means, if $(1,i)$ is the strategy corresponding to $p^1_i$, then $w(v,(1,i)|\tilde{\uu})\leq  w(v,1|\uu)\leq w(v,c|\tilde{\uu})$ for all $i$. Of course, it is still true that $ w(v,c|\tilde{\uu})\geq w(v,s|\tilde{\uu})$ for any $s\neq 1$ because these vertices are unchanged from the original Nash equilibrium. Thus $v$ is playing its best response in $\tilde{\uu}$.

    Now we consider those vertices $v\in p^1\setminus\partial p^1$. Notice that under the refinement, $v\in p^1_i\setminus\partial p^1_i$ because the refinement only separates the connected components of the subgraph spanned by $p^1$. Of course, those vertices which only neighbor other vertices using the same strategy are using their best response. This continues to be true under the refinement.

    Lastly, we consider those vertices $v\in \partial p^1$ which, under the refinement, are in $\partial p^1_i$ for some $i$. $\Gamma(v)\cap p^1=\Gamma(v)\cap P^1_i$ because the refinement did not break any connected components of $p^1$ into separate parts. Thus, $w(v,1|\uu)=w(v,(1,i)|\tilde{\uu})$. Because $v$ was playing a best response with respect to $\uu$, $w(v,1|\uu)\geq w(v,c|\uu)=w(v,c|\tilde{\uu})$ and because $p^1_i$ and $p^1_j$ have no neighboring vertices for any $i\neq j$, $w(v,(1,j)|\tilde{\uu})=0$ for any $i \neq j$. Therefore, $w(v,(1,i)|\tilde{\uu})\geq w(v,c|\tilde{\uu})$ for any $c>1$ or $c = (1,1),...(1,n)$. Thus, $v$ is playing a best response under the refinement. 

    Therefore, under the refinement, every vertex is still using a best response strategy, so $\tilde{\uu}$ is a Nash equilibrium and the refinement is an equilibrium partition.   
\end{proof}

While purpose of the present paper is to give a numerical treatment of the system in order to inform more general analytical results, some analytical results can be shown for simple graphs with nice properties. Theorems \ref{KnTheorem} and \ref{KnnTheorem} are two such results which will be referred to later in the discussion as we make conjectures about more general results. 

\begin{theorem}\label{KnTheorem} $K_n$ is indecomposable.
		\end{theorem}
		\begin{proof} 
		Suppose, by way of contradiction, there is an equilibrium strategy profile $\mathbf{u^*}$ with $d\geq2$ clusters, $q^{c_1},...,q^{c_d}$. Because every pair of vertices shares an edge,
			\begin{equation}\label{Knfitness}
				w(v,c_i|\uu)=\begin{cases}
					|q^{c_i}|&\uu_v=c_i\\
					|q^{c_i}|-1&\uu_v\neq c_i
				\end{cases}
			\end{equation} so the fitness of a vertex is directly associated with the size of the cluster of which it is a part.  Consider a vertex $v_1\in q^{c_1}$. Because it is at equilibrium, $$w(v_1,c_1|\uu^*)=\max_C\{w(v_1,c|\uu^*)\}=:a.$$
			Therefore, $|q^{c_1}|=a+1$. Moreover, If $v_1$ played a different strategy, $c_2$, its fitness would be  $w(v_1,c_2|\uu^*)\leq a$ so $|q^{c_2}|\leq a$.
			Now consider a separate vertex, $v_2\in q^{c_2}$. $w(v_2,c_2|\uu^*)=|q^{c_2}|-1\leq a-1$ by  equation \eqref{Knfitness}. However, $w(v_2,c_1|\uu^*)=a+1$. 
			Thus $\uu^*_{v_2}=c_2\notin \argmax\{w(v_2,c|\uu^*)\}$, so $\uu^*$ is not an equilibrium.  
		\end{proof}

	\begin{theorem}\label{KnnTheorem}{$K_{n,m}$ admits an equilibrium partition with $d$ parts if and only if $d|n$ and $d|m$.}
		\end{theorem}
	
		\begin{proof} Consider the complete bipartite graph $K_{n,m}$ which has parts $E^n$ and $E^m$, with $n$ and $m$ vertices, respectively.

			$\impliedby$ Let $d$ be a common divisor of $m$ and $n$. It is sufficient to construct an equilibrium strategy profile with $d$ strategies. Partition the vertices of $E^m$ so that there are $m/d$ vertices using each strategy. Likewise, partition the vertices of $E^n$ so that there are $n/d$ vertices using each strategy. Consider $v_1\in E^m$.  Notice, because it shares an edge with every vertex in $E^n$, its fitness is certainly 
			\begin{equation}
				w(v_1,c|\uu)=\frac{n}{d}\quad \forall\, c\in C
			\end{equation}
			Therefore, it is certainly playing a best response. The same can be said for every vertex in $E^n$. Thus we have constructed an equilibrium strategy profile with $d$ clusters and so the corresponding partition is an equilibrium partition.

			$\implies$ Suppose, by way of contradiction   and Without loss of generality, that $d$ is not a divisor of $m$. Suppose further that $\uu^*$ is an equilibrium strategy profile with $d$ clusters. $d$ does not divide $m$, so $\exists$ strategies $r$ and $s$ such that $|q^r\cap E^m|>|q^s\cap E^m|$. If $\uu^*$ is at equilibrium, it must follow that $q^s\cap E^n=\emptyset$. It then follows that $q^s\cap E^m =\emptyset$. If $q^s$ is empty, then there are not $d$ clusters in $\uu^*$. This contradiction proves that, if there is an equilibrium strategy profile in $K_{n,m}$, the number of clusters must be a common divisor of $n$ and $m$.   
		\end{proof}
    \begin{corollary}
        $K_{n,m}$ is indecomposable if and only if $n$ and $m$ are coprime. 
    \end{corollary}

The results in both theorems \ref{KnTheorem} and \ref{KnnTheorem} depend heavily on the structure of the graph. For this reason, we cannot extrapolate these techniques to draw general conclusions about the system. Instead, we start by building up a numerical understanding of the system so that we can get a better idea of the types of behavior we can expect. 
 
\section{Catalogue of Equilibrium Partitions on Small Graphs} \label{Catalogue}
In order to start working towards more robust analytical results, we start by presenting a catalogue of all the equilibrium partitions on small graphs. We will see that, for the smallest graphs, non-trivial partitions are rare, but even by $n=6$, less than half of connected graphs are indecomposable (and note obviously that all disconnected graphs are decomposable simply by partitioning connected components). This catalogue domonstrates that indecomposability depends highly on the exact structure of the graph, and to support the conjecture that indecomposability will become asymptotically rare in Erd\H{o}s-R\'enyi random graphs ($\Gamma_{n,N(n)}$) \cite{Erdos1959} as $n\rightarrow \infty$.

Although it is very easy to make an exhaustive list of equilibrium partitions on graphs with no more than 5 vertices using only the results from section \ref{Background}, when the number of graphs to consider grows, we resort to computer assisted exhaustive search. First note that it is a very easy computation to check if a partition is an equilibrium partition. For a graph with $n$ vertices the operation runs in $\mathcal{O}(n^2)$, but we consider only small $n$, so this step is practically trivial. With an easily available catalogue of connected graphs of order at most seven, it is a nearly trivial task to consider every possible vertex partition. For ease of computation, we can disregard any partition which has a part of size 1 by lemma \ref{clusters of size 1}. With far fewer than Bells number, the number of partitions of $n$ distinguishable elements, $B(n)$,\cite{AWalkThroughCombinatorics} of partitions to check for each graph of order $n$, the exhaustive search can be completed very quickly. 

The biggest complication is making sure that each partition is unique up to isomorphism; this is non-trivial step but it can be done, albeit inelegantly,  without much elegance, by relying on algorithms to detect graph isomorphisms. 

There are three ways by which two partitions of the same graph might be isomorphic. The first, and easier to detect, is the case in which the parts of the partitions are simply relabeled (Fig.\ref{isofigure} row 1). We can exclude many of these isomorphisms by relabelling early in the algorithm without much difficulty. 

The next way for two partitions of the same graph to be isomorphic is by way of graphical symmetry(Fig. \ref{isofigure} row 2). To tackle this problem we need only find an injection, $\Psi$, from the set of labeled partitions of labeled graph to the set of labeled graphs. We do this through graph expansion. By adding a number of pendant vertices of each vertex in the original graph corresponding to the part of the partition to which it belongs (i.e. if $P=\{q_i\}_{i=1}^n$, then $v\in q_i$ is given $i$ pendent vertices) we have such an injection. If partitions are isomorphic by way of graphical symmetry then they will be expanded to labeled graphs which are isomorphic. $P_1\simeq P_2\Rightarrow \Psi(P_1)\sim\Psi(P_2)$. More specifically, the labeled graphs will be isomorphic to one another, so as unlabeled graphs they are identical. Moreover, because this expansion is injective, if two expanded graphs are identical (or the two labeled graphs are isomorphic) to one another, it is certain that the corresponding graph partitions are isomporphic. $\Psi(P_1)\sim\Psi(P_2)\Rightarrow P_1\simeq P_2$.  There are existing efficient algorithms to detect isomorphisms between labeled graphs, so we take advantage of these to efficiently detect isomorphisms between graph partitions. 

The final way two partitions might be isomorphic is through a combination of graphical symmetry and relabelling (Fig. \ref{isofigure} row 3). In this case, because we deal with very small partitions, we can simply check for isomorphism through graphical symmetry for every possible relabelling. Note that it is far more efficient to initially only test the partitions which are distinct up to relabelling, then add in the relabelled versions of the partitions to test in this third case, than it is to run the algorithm without any removal. 

\begin{figure*}[h!]
    \centering
    \begin{tabular}{|m{2cm}|m{5.5cm}|c|}
    \hline
        Type& Partition Diagram& Vector View \\
         \hline
         \hline
         \raggedright
         Isomorphic by way of relabeling & \begin{tikzpicture}
            \node(a)[circle, fill, inner sep =1.5pt] at (0,0){};
            \node(b)[circle, fill, inner sep = 1.5pt] at(0.5,-0.5){};
            \node(c)[circle, fill, inner sep = 1.5pt] at(1,0){};
            \node(d)[circle, fill, inner sep = 1.5pt] at(1.5,-0.5){};
            \node(e)[circle, fill, inner sep = 1.5pt] at(2,0){};

            \draw (a)--(b)--(c)--(d)--(e);
            
		  \node[fit=(a)(b)(c),dashed, draw, rectangle,rounded corners=10,inner sep=5pt] {};
        \node[fit=(d)(e),dashed, draw, rectangle,rounded corners=10,inner sep=5pt] {};

            \node(label11) at (0.5,0.5){$p^1$};
            \node(label12) at (1.75,0.5){$p^2$};

            \node(equal)at (2.5,-0.25){$\simeq$};

            \node(f)[circle, fill, inner sep =1.5pt] at (3,0){};
            \node(g)[circle, fill, inner sep = 1.5pt] at(3.5,-0.5){};
            \node(h)[circle, fill, inner sep = 1.5pt] at(4,0){};
            \node(i)[circle, fill, inner sep = 1.5pt] at(4.5,-0.5){};
            \node(j)[circle, fill, inner sep = 1.5pt] at(5,0){};

            \draw (f)--(g)--(h)--(i)--(j);
            
		  \node[fit=(f)(g)(h),dashed, draw, rectangle,rounded corners=10,inner sep=5pt] {};
        \node[fit=(i)(j),dashed, draw, rectangle,rounded corners=10,inner sep=5pt] {};

            \node(label21) at (3.5,0.5){$p^2$};
            \node(label22) at (4.75,0.5){$p^1$};
            
         \end{tikzpicture} 
         & $[1,1,1,2,2]\simeq[2,2,2,1,1]$\\
         \hline \raggedright
         Isomorphic by way of symmerty & \begin{tikzpicture}
            \node(a)[circle, fill, inner sep =1.5pt] at (0,0){};
            \node(b)[circle, fill, inner sep = 1.5pt] at(0.5,-0.5){};
            \node(c)[circle, fill, inner sep = 1.5pt] at(1,0){};
            \node(d)[circle, fill, inner sep = 1.5pt] at(1.5,-0.5){};
            \node(e)[circle, fill, inner sep = 1.5pt] at(2,0){};

            \draw (a)--(b)--(c)--(d)--(e);
            
		  \node[fit=(a)(b)(c),dashed, draw, rectangle,rounded corners=10,inner sep=5pt] {};
        \node[fit=(d)(e),dashed, draw, rectangle,rounded corners=10,inner sep=5pt] {};

            \node(label11) at (0.5,0.5){$p^1$};
            \node(label12) at (1.75,0.5){$p^2$};

            \node(equal)at (2.5,-0.25){$\simeq$};

            \node(f)[circle, fill, inner sep =1.5pt] at (3,0){};
            \node(g)[circle, fill, inner sep = 1.5pt] at(3.5,-0.5){};
            \node(h)[circle, fill, inner sep = 1.5pt] at(4,0){};
            \node(i)[circle, fill, inner sep = 1.5pt] at(4.5,-0.5){};
            \node(j)[circle, fill, inner sep = 1.5pt] at(5,0){};

            \draw (f)--(g)--(h)--(i)--(j);
            
		  \node[fit=(f)(g),dashed, draw, rectangle,rounded corners=10,inner sep=5pt] {};
        \node[fit=(h)(i)(j),dashed, draw, rectangle,rounded corners=10,inner sep=5pt] {};

            \node(label21) at (3.25,0.5){$p^2$};
            \node(label22) at (4.5,0.5){$p^1$};
            
         \end{tikzpicture} 
         & $[1,1,1,2,2]\simeq[2,2,1,1,1]$\\
         \hline\raggedright
         Isomorphic by way of symmetry and relabeling & \begin{tikzpicture}
            \node(a)[circle, fill, inner sep =1.5pt] at (0,0){};
            \node(b)[circle, fill, inner sep = 1.5pt] at(0.5,-0.5){};
            \node(c)[circle, fill, inner sep = 1.5pt] at(1,0){};
            \node(d)[circle, fill, inner sep = 1.5pt] at(1.5,-0.5){};
            \node(e)[circle, fill, inner sep = 1.5pt] at(2,0){};

            \draw (a)--(b)--(c)--(d)--(e);
            
		  \node[fit=(a)(b)(c),dashed, draw, rectangle,rounded corners=10,inner sep=5pt] {};
        \node[fit=(d)(e),dashed, draw, rectangle,rounded corners=10,inner sep=5pt] {};

            \node(label11) at (0.5,0.5){$p^1$};
            \node(label12) at (1.75,0.5){$p^2$};

            \node(equal)at (2.5,-0.25){$\simeq$};

            \node(f)[circle, fill, inner sep =1.5pt] at (3,0){};
            \node(g)[circle, fill, inner sep = 1.5pt] at(3.5,-0.5){};
            \node(h)[circle, fill, inner sep = 1.5pt] at(4,0){};
            \node(i)[circle, fill, inner sep = 1.5pt] at(4.5,-0.5){};
            \node(j)[circle, fill, inner sep = 1.5pt] at(5,0){};

            \draw (f)--(g)--(h)--(i)--(j);
            
		  \node[fit=(f)(g),dashed, draw, rectangle,rounded corners=10,inner sep=5pt] {};
        \node[fit=(h)(i)(j),dashed, draw, rectangle,rounded corners=10,inner sep=5pt] {};

            \node(label21) at (3.25,0.5){$p^1$};
            \node(label22) at (4.5,0.5){$p^2$};
            
         \end{tikzpicture} 
         & $[1,1,1,2,2]\simeq[1,1,2,2,2]$\\
         \hline
    \end{tabular}
    \caption{An example of the three different ways two labeled partitions of labeled graphs can be isomorphic to one another and the ``vector view," which is simply how the computer stores the partition information. The first row shows two partitions which are isomorphic by way of relabelling. The second row shows two partitions which are isomorphic by way of graphical symmetry. The third row shows two partitions which are isomorphic by way of both graphical symmetry and relabeling. }
    \label{isofigure}
\end{figure*}

In the very worst case, the algorithm to detect partition isomorphisms has complexity $n!$ times the complexity of the graph isomorphism detection algorithm ($2^{\mathcal{O}(\log(n)^3}$)
\cite{helfgott2017graph}
but through use of lemma \ref{clusters of size 1} and early removal of relabelled partitions, we can get this below $\left\lfloor\frac{n}{2}\right\rfloor!2^{\mathcal{O}(\log(4n)^3)}$, which is more than sufficient for the size of graphs we will catalogue. Notice the replacement of $n$ with $4n$ is due to the fact that our injection adds up to $\lfloor\frac{n}{2}\rfloor n $ vertices to the graph.

The only remaining complication is whether we include disconnected parts in our partitions. For our purposes, because of lemma \ref{connected clusters}, we will only consider those partitions with connected parts. Therefore, we discard any partitions with disconnected parts and do not include them in our catalogue. The refinement of such a partition, wherein every part has a connected spanning subgraph, will certainly be included.  

Having resolved both of these complications, we can examine every graph with at most 7 vertices and do an exhaustive search of all equilibrium partitions. In the following subsections, the results of the search are presented.

\subsection{n=1}
Finding equilibrium partitions in $K_1$ is trivial (Fig. \ref{n1equilibria }).
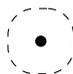
\begin{figure}[h!]
    \centering
        \begin{tikzpicture}
			\node(a)[circle, fill, inner sep =1.5pt] at (1.5,0){};
		
			\node[fit=(a),dashed, draw, rectangle,rounded corners=10,inner sep=10pt] {};	
		\end{tikzpicture}

    \caption{The only equilibrium partition in $K_1$ is obviously the trivial partition. There is no further investigation required.}
    \label{n1equilibria }
\end{figure}

\subsection{n=2}
Finding equilibrium partitions in $K_2$ is also made trivial by lemma \ref{clusters of size 1} (Fig.\ref{n2equilibria}).
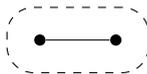
\begin{figure}[h!]
    \centering 
         \begin{tikzpicture}
			\node(a)[circle, fill, inner sep =1.5pt] at (0,0){};
            \node(b)[circle, fill, inner sep = 1.5pt] at(1,0){};

            \draw (a)--(b);
            
			\node[fit=(a)(b),dashed, draw, rectangle,rounded corners=10,inner sep=10pt] {};	
		\end{tikzpicture}
  \caption{The only equilibrium in $K_2$ is the trivial partition}
  \label{n2equilibria}
\end{figure}

\subsection{n=3}
There are only two connected, non-isomorphic graphs of order 3 and the only equilibrium partition admitted by either is, again, the trivial partition (Fig. \ref{n3equilibria}). This is also immediate from lemma \ref{clusters of size 1}. 

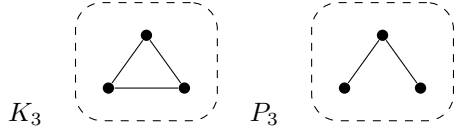
\begin{figure}[h!]
    \centering 
    \begin{tabular}{rcrc}
    \textbf{$K_3$}&
        \begin{tikzpicture}
		  \node(a)[circle, fill, inner sep =1.5pt] at (0,0){};
            \node(b)[circle, fill, inner sep = 1.5pt] at(1,0){};
            \node(c)[circle, fill, inner sep = 1.5pt] at(0.5,0.7){};

            \draw (a)--(b)--(c)--(a);
            
		  \node[fit=(a)(b)(c),dashed, draw, rectangle,rounded corners=10,inner sep=10pt] {};	
		\end{tikzpicture}
  & \textbf{$P_3$}&\begin{tikzpicture}
		  \node(a)[circle, fill, inner sep =1.5pt] at (0,0){};
            \node(b)[circle, fill, inner sep = 1.5pt] at(1,0){};
            \node(c)[circle, fill, inner sep = 1.5pt] at(0.5,0.7){};

            \draw (a)--(c)--(b);
            
		  \node[fit=(a)(b)(c),dashed, draw, rectangle,rounded corners=10,inner sep=10pt] {};	
		\end{tikzpicture}
  \end{tabular}
  \caption{The only equilibrium in $K_3$ and $P_3$ is the trivial partition}
  \label{n3equilibria}
\end{figure}

\subsection{n=4}
There are 6 connected, non-isomorphic graphs of order 4. Four of them are indecomposable (Fig. \ref{indecomposable4}). By lemma \ref{clusters of size 1}, any candidate equilibrium partition must have two parts, each with two vertices. With this fact, there are very few potential partitions to consider and each can be ruled out using the definition of equilibrium. 

\begin{figure}[h!]
    \centering
    \begin{tabular}{cccc}
         $K_4$& \begin{tikzpicture}
		  \node(a)[circle, fill, inner sep =1.5pt] at (0,0){};
            \node(b)[circle, fill, inner sep = 1.5pt] at(1,0){};
            \node(c)[circle, fill, inner sep = 1.5pt] at(0,1){};
            \node(d)[circle, fill, inner sep = 1.5pt] at(1,1){};

            \draw (a)--(c)--(b)--(d)--(a)--(b);
            \draw (c)--(d);	
		\end{tikzpicture}& $K_4-e$&
         \begin{tikzpicture}
		  \node(a)[circle, fill, inner sep =1.5pt] at (0,0){};
            \node(b)[circle, fill, inner sep = 1.5pt] at(1,0){};
            \node(c)[circle, fill, inner sep = 1.5pt] at(0,1){};
            \node(d)[circle, fill, inner sep = 1.5pt] at(1,1){};

            \draw (a)--(c)--(b)--(d)--(a)--(b);
            	
		\end{tikzpicture}\\
        $3-pan$&\begin{tikzpicture}
		  \node(a)[circle, fill, inner sep =1.5pt] at (0,0){};
            \node(b)[circle, fill, inner sep = 1.5pt] at(1,0){};
            \node(c)[circle, fill, inner sep = 1.5pt] at(0,1){};
            \node(d)[circle, fill, inner sep = 1.5pt] at(1,1){};

            \draw (a)--(c)--(b)--(d);
            \draw (a)--(b);

		\end{tikzpicture}&
        $K_{1,3}$&\begin{tikzpicture}
		  \node(a)[circle, fill, inner sep =1.5pt] at (0,0){};
            \node(b)[circle, fill, inner sep = 1.5pt] at(1,0){};
            \node(c)[circle, fill, inner sep = 1.5pt] at(0,1){};
            \node(d)[circle, fill, inner sep = 1.5pt] at(1,1){};

            \draw (a)--(b);
            \draw (a)--(c);
            \draw (a)--(d);

		\end{tikzpicture}\\
  
    \end{tabular}
    \caption{The indecomposable connected graphs on four vertices}
    \label{indecomposable4}
\end{figure}
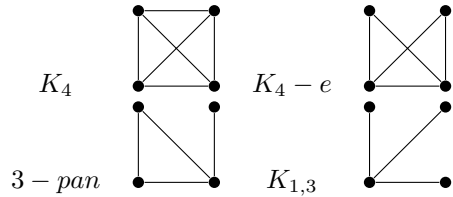

The remaining two connected graphs have non-trivial equilibrium partitions. In addition to the trivial partition, they each have only one non-trivial partition (Fig. \ref{decomposable 4}).
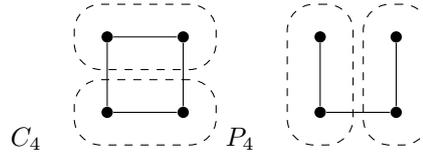
\begin{figure}[h!]
    \centering
    \begin{tabular}{cccc}
         $C_4$&\begin{tikzpicture}
		  \node(a)[circle, fill, inner sep =1.5pt] at (0,0){};
            \node(b)[circle, fill, inner sep = 1.5pt] at(1,0){};
            \node(c)[circle, fill, inner sep = 1.5pt] at(0,1){};
            \node(d)[circle, fill, inner sep = 1.5pt] at(1,1){};

            \draw (a)--(b)--(d)--(c)--(a);
            
		  \node[fit=(a)(b),dashed, draw, rectangle,rounded corners=10,inner sep=10pt] {};
            \node[fit=(c)(d),dashed, draw, rectangle,rounded corners=10,inner sep=10pt] {};
		\end{tikzpicture}
         $P_4$&\begin{tikzpicture}
		  \node(a)[circle, fill, inner sep =1.5pt] at (0,0){};
            \node(b)[circle, fill, inner sep = 1.5pt] at(1,0){};
            \node(c)[circle, fill, inner sep = 1.5pt] at(0,1){};
            \node(d)[circle, fill, inner sep = 1.5pt] at(1,1){};

            \draw (a)--(b);
            \draw (a)--(c);
            \draw (b)--(d);
            
		  \node[fit=(a)(c),dashed, draw, rectangle,rounded corners=10,inner sep=10pt] {};
            \node[fit=(b)(d),dashed, draw, rectangle,rounded corners=10,inner sep=10pt] {};
		\end{tikzpicture}\\
    \end{tabular}
    \caption{The two connected graphs on four vertices which are decomposable. }
    \label{decomposable 4}
\end{figure}
Notice that, if the vertices of $C_4$ are labeled, there are two distinct equilibrium partitions, however, they are clearly isomorphic to one another and so only count as one equilibrium partition.  

\subsection{n=5}
There are 21 connected non-isomorphic graphs of order 5. It is easy to see that every equilibrium on these graphs has either one or two clusters, again by lemma \ref{clusters of size 1}. This leaves only $11$ candidate partitions to consider for each connected graph. Of the 21 connected graphs, 13 of them are indecomposable (Fig. \ref{n5indecomposable}).

\begin{figure*}[h!]
    \centering
    \begin{tabular}{cccccc}
         $K_5$& \begin{tikzpicture}
		  \node(a)[circle, fill, inner sep =1.5pt] at (0,0){};
            \node(b)[circle, fill, inner sep = 1.5pt] at(1,0){};
            \node(c)[circle, fill, inner sep = 1.5pt] at(-0.2,1){};
            \node(d)[circle, fill, inner sep = 1.5pt] at(1.2,1){};
            \node(e)[circle, fill, inner sep = 1.5pt] at(0.5,1.6){};

            \draw (e)--(a)--(c)--(b)--(d)--(a)--(b)--(e);
            \draw (e)--(c)--(d)--(e);
            
		\end{tikzpicture} &$K_5-e$& \begin{tikzpicture}
		    \node(a)[circle, fill, inner sep =1.5pt] at (0,0){};
            \node(b)[circle, fill, inner sep = 1.5pt] at(1,0){};
            \node(c)[circle, fill, inner sep = 1.5pt] at(-0.2,1){};
            \node(d)[circle, fill, inner sep = 1.5pt] at(1.2,1){};
            \node(e)[circle, fill, inner sep = 1.5pt] at(0.5,1.6){};

            \draw (e)--(a)--(c)--(b)--(d)--(a);
            \draw (b)--(e);
            \draw (e)--(c)--(d)--(e);
		\end{tikzpicture}&
    $\overline{P_3\cup 2K_1}$&
    \begin{tikzpicture}
		    \node(a)[circle, fill, inner sep =1.5pt] at (0,0){};
            \node(b)[circle, fill, inner sep = 1.5pt] at(1,0){};
            \node(c)[circle, fill, inner sep = 1.5pt] at(-0.2,1){};
            \node(d)[circle, fill, inner sep = 1.5pt] at(1.2,1){};
            \node(e)[circle, fill, inner sep = 1.5pt] at(0.5,1.6){};

            \draw (a)--(c)--(b)--(d)--(a)--(b)--(e);
            \draw (c)--(d);
            \draw (a)--(e);
		\end{tikzpicture}\\

        $W_4$& \begin{tikzpicture}
		  \node(a)[circle, fill, inner sep =1.5pt] at (0,0){};
            \node(b)[circle, fill, inner sep = 1.5pt] at(1,0){};
            \node(c)[circle, fill, inner sep = 1.5pt] at(-0.2,1){};
            \node(d)[circle, fill, inner sep = 1.5pt] at(1.2,1){};
            \node(e)[circle, fill, inner sep = 1.5pt] at(0.5,1.6){};

            \draw (e)--(a)--(c);
            \draw (b)--(d);
            \draw (a)--(b)--(e);
            \draw (e)--(c)--(d)--(e);
            
		\end{tikzpicture} & $\overline{Claw \cup K_1}$& \begin{tikzpicture}
		    \node(a)[circle, fill, inner sep =1.5pt] at (0,0){};
            \node(b)[circle, fill, inner sep = 1.5pt] at(1,0){};
            \node(c)[circle, fill, inner sep = 1.5pt] at(-0.2,1){};
            \node(d)[circle, fill, inner sep = 1.5pt] at(1.2,1){};
            \node(e)[circle, fill, inner sep = 1.5pt] at(0.5,1.6){};

            \draw (a)--(c)--(b)--(d)--(a)--(b);
            \draw (e)--(c)--(d);
            \draw (b)--(c);
		\end{tikzpicture}&
    $\overline{P_3\cup P_2}$&
    \begin{tikzpicture}
		    \node(a)[circle, fill, inner sep =1.5pt] at (0,0){};
            \node(b)[circle, fill, inner sep = 1.5pt] at(1,0){};
            \node(c)[circle, fill, inner sep = 1.5pt] at(-0.2,1){};
            \node(d)[circle, fill, inner sep = 1.5pt] at(1.2,1){};
            \node(e)[circle, fill, inner sep = 1.5pt] at(0.5,1.6){};

            \draw (a)--(c)--(b)--(d)--(a);
            \draw (b)--(e);
            \draw (c)--(d);
            \draw (a)--(e);
		\end{tikzpicture}\\
        $3-fan$& \begin{tikzpicture}
		  \node(a)[circle, fill, inner sep =1.5pt] at (0,0){};
            \node(b)[circle, fill, inner sep = 1.5pt] at(1,0){};
            \node(c)[circle, fill, inner sep = 1.5pt] at(-0.2,1){};
            \node(d)[circle, fill, inner sep = 1.5pt] at(1.2,1){};
            \node(e)[circle, fill, inner sep = 1.5pt] at(0.5,1.6){};

            \draw (e)--(a)--(c);
            \draw (b)--(d);
            \draw (a)--(b)--(e);
            \draw (e)--(c);
            \draw (d)--(e);
            
		\end{tikzpicture} & $\overline{K_3 \cup 2K_1}$& \begin{tikzpicture}
		    \node(a)[circle, fill, inner sep =1.5pt] at (0,0){};
            \node(b)[circle, fill, inner sep = 1.5pt] at(1,0){};
            \node(c)[circle, fill, inner sep = 1.5pt] at(-0.2,1){};
            \node(d)[circle, fill, inner sep = 1.5pt] at(1.2,1){};
            \node(e)[circle, fill, inner sep = 1.5pt] at(0.5,1.6){};

            \draw (a)--(c)--(b)--(d)--(a);
            \draw (e)--(c)--(d);
            \draw (e)--(d);
            \draw (b)--(c);
		\end{tikzpicture}&
    $K_{1,4}$&
    \begin{tikzpicture}
		    \node(a)[circle, fill, inner sep =1.5pt] at (0,0){};
            \node(b)[circle, fill, inner sep = 1.5pt] at(1,0){};
            \node(c)[circle, fill, inner sep = 1.5pt] at(-0.2,1){};
            \node(d)[circle, fill, inner sep = 1.5pt] at(1.2,1){};
            \node(e)[circle, fill, inner sep = 1.5pt] at(0.5,1.6){};

            \draw (a)--(e);
            \draw (b)--(e);
            \draw (c)--(e);
            \draw (d)--(e);
		\end{tikzpicture}\\

    $\overline{fork}$& \begin{tikzpicture}
		  \node(a)[circle, fill, inner sep =1.5pt] at (0,0){};
            \node(b)[circle, fill, inner sep = 1.5pt] at(1,0){};
            \node(c)[circle, fill, inner sep = 1.5pt] at(-0.2,1){};
            \node(d)[circle, fill, inner sep = 1.5pt] at(1.2,1){};
            \node(e)[circle, fill, inner sep = 1.5pt] at(0.5,1.6){};

            \draw (a)--(c);
            \draw (b)--(d)--(a);
            \draw (a)--(b);
            \draw (c)--(d);
            \draw (e)--(c);
            
		\end{tikzpicture} & $dart$& \begin{tikzpicture}
		    \node(a)[circle, fill, inner sep =1.5pt] at (0,0){};
            \node(b)[circle, fill, inner sep = 1.5pt] at(1,0){};
            \node(c)[circle, fill, inner sep = 1.5pt] at(-0.2,1){};
            \node(d)[circle, fill, inner sep = 1.5pt] at(1.2,1){};
            \node(e)[circle, fill, inner sep = 1.5pt] at(0.5,1.6){};

            \draw (a)--(c)--(b);
            \draw (b)--(d);
            \draw (a)--(b);
            \draw (c)--(d);
            \draw (e)--(c);
		\end{tikzpicture}&
    $K_{2,3}$&
    \begin{tikzpicture}
		    \node(a)[circle, fill, inner sep =1.5pt] at (0,0){};
            \node(b)[circle, fill, inner sep = 1.5pt] at(1,0){};
            \node(c)[circle, fill, inner sep = 1.5pt] at(-0.2,1){};
            \node(d)[circle, fill, inner sep = 1.5pt] at(1.2,1){};
            \node(e)[circle, fill, inner sep = 1.5pt] at(0.5,1.6){};

            \draw (a)--(c)--(b)--(d)--(a);
            \draw (e)--(c);
            \draw (e)--(d);
            \draw (b)--(c);
		\end{tikzpicture}\\
        &&bull & \begin{tikzpicture}
		    \node(a)[circle, fill, inner sep =1.5pt] at (0,0){};
            \node(b)[circle, fill, inner sep = 1.5pt] at(1,0){};
            \node(c)[circle, fill, inner sep = 1.5pt] at(-0.2,1){};
            \node(d)[circle, fill, inner sep = 1.5pt] at(1.2,1){};
            \node(e)[circle, fill, inner sep = 1.5pt] at(0.5,1.6){};

            \draw (a)--(c);
            \draw (b)--(d)--(c);
            \draw (e)--(c);
            \draw (e)--(d);
		\end{tikzpicture}&&
         
    \end{tabular}
    \caption{Every indecomposable connected graph on 5 vertices.}
    \label{n5indecomposable}
\end{figure*}
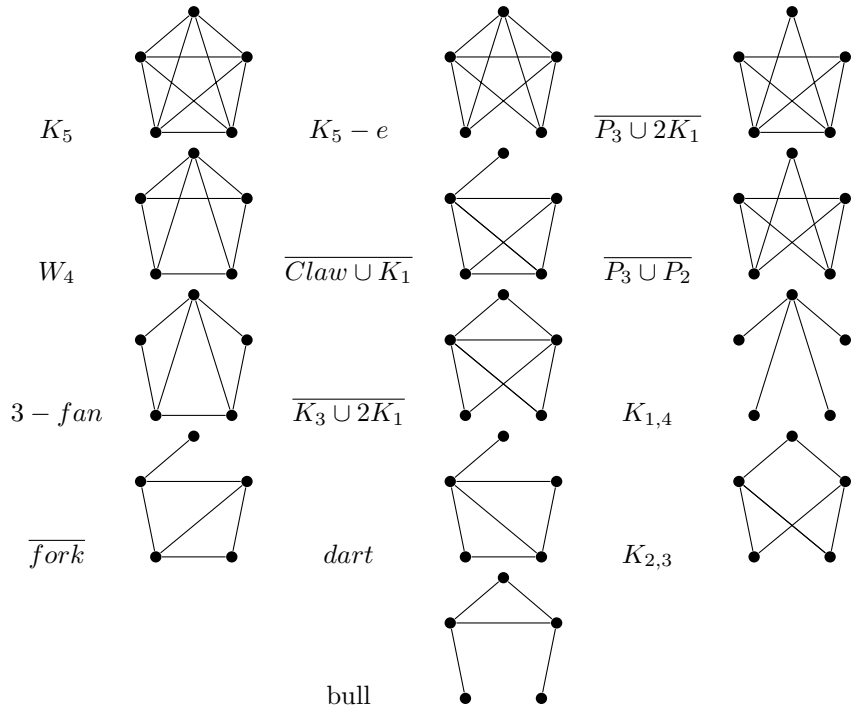

The remaining 8 are decomposable. Each obviously admits the trivial partition, as well as a unique non-trivial partition, up to isomorphism. Each of the non-trivial equilibrium partitions are pictured in Figure \ref{n5decomposable}.

\begin{figure*}[h!]
    \centering
    \begin{tabular}{cccc}
         $butterfly$& \begin{tikzpicture}
		    \node(a)[circle, fill, inner sep =1.5pt] at (0,0){};
            \node(b)[circle, fill, inner sep = 1.5pt] at(2,0){};
            \node(c)[circle, fill, inner sep = 1.5pt] at(0,1){};
            \node(d)[circle, fill, inner sep = 1.5pt] at(2,1){};
            \node(e)[circle, fill, inner sep = 1.5pt] at(1,0.5){};

            \draw (a)--(e)--(c)--(a);
            \draw (b)--(e)--(d)--(b);

            \node[fit=(a)(c)(e),dashed, draw, rectangle,rounded corners=10,inner sep=10pt] {};
            \node[fit=(b)(d),dashed, draw, rectangle,rounded corners=10,inner sep=10pt] {};
		\end{tikzpicture}&
    $fork$ &\begin{tikzpicture}
		    \node(a)[circle, fill, inner sep =1.5pt] at (0,0){};
            \node(b)[circle, fill, inner sep = 1.5pt] at(2,0){};
            \node(c)[circle, fill, inner sep = 1.5pt] at(0,1){};
            \node(d)[circle, fill, inner sep = 1.5pt] at(2,1){};
            \node(e)[circle, fill, inner sep = 1.5pt] at(1,0.5){};

            \draw (a)--(e)--(c);
            \draw (d)--(b)--(e);

            \node[fit=(a)(c)(e),dashed, draw, rectangle,rounded corners=10,inner sep=10pt] {};
            \node[fit=(b)(d),dashed, draw, rectangle,rounded corners=10,inner sep=10pt] {};
		\end{tikzpicture} \\

     $P_5$& \begin{tikzpicture}
		    \node(a)[circle, fill, inner sep =1.5pt] at (0,0){};
            \node(b)[circle, fill, inner sep = 1.5pt] at(2,0){};
            \node(c)[circle, fill, inner sep = 1.5pt] at(0,1){};
            \node(d)[circle, fill, inner sep = 1.5pt] at(2,1){};
            \node(e)[circle, fill, inner sep = 1.5pt] at(1,0.5){};

            \draw (e)--(c)--(a);
            \draw (d)--(b)--(e);

            \node[fit=(a)(c)(e),dashed, draw, rectangle,rounded corners=10,inner sep=10pt] {};
            \node[fit=(b)(d),dashed, draw, rectangle,rounded corners=10,inner sep=10pt] {};
		\end{tikzpicture}&
    $\overline{P_5}$ &\begin{tikzpicture}
		\node(a)[circle, fill, inner sep =1.5pt] at (0,0){};
            \node(b)[circle, fill, inner sep = 1.5pt] at(2,0){};
            \node(c)[circle, fill, inner sep = 1.5pt] at(0,1){};
            \node(d)[circle, fill, inner sep = 1.5pt] at(2,1){};
            \node(e)[circle, fill, inner sep = 1.5pt] at(1,0.5){};

            \draw (a)--(e);
            \draw (e)--(c)--(a)--(b);
            \draw (c)--(d)--(b);

            \node[fit=(a)(c)(e),dashed, draw, rectangle,rounded corners=10,inner sep=10pt] {};
            \node[fit=(b)(d),dashed, draw, rectangle,rounded corners=10,inner sep=10pt] {};
		\end{tikzpicture} \\

    $4-pan$& \begin{tikzpicture}
		    \node(a)[circle, fill, inner sep =1.5pt] at (0,0){};
            \node(b)[circle, fill, inner sep = 1.5pt] at(2,0){};
            \node(c)[circle, fill, inner sep = 1.5pt] at(0,1){};
            \node(d)[circle, fill, inner sep = 1.5pt] at(2,1){};
            \node(e)[circle, fill, inner sep = 1.5pt] at(1,0.5){};

            \draw (e)--(c)--(a);
            \draw (c)--(d)--(b)--(e);

            \node[fit=(a)(c)(e),dashed, draw, rectangle,rounded corners=10,inner sep=10pt] {};
            \node[fit=(b)(d),dashed, draw, rectangle,rounded corners=10,inner sep=10pt] {};
		\end{tikzpicture}&
    $co-4-pan$ &\begin{tikzpicture}
		\node(a)[circle, fill, inner sep =1.5pt] at (0,0){};
            \node(b)[circle, fill, inner sep = 1.5pt] at(2,0){};
            \node(c)[circle, fill, inner sep = 1.5pt] at(0,1){};
            \node(d)[circle, fill, inner sep = 1.5pt] at(2,1){};
            \node(e)[circle, fill, inner sep = 1.5pt] at(1,0.5){};

            \draw (a)--(e);
            \draw (e)--(c)--(a)--(b);
            \draw (d)--(b);

            \node[fit=(a)(c)(e),dashed, draw, rectangle,rounded corners=10,inner sep=10pt] {};
            \node[fit=(b)(d),dashed, draw, rectangle,rounded corners=10,inner sep=10pt] {};
		\end{tikzpicture} \\

    $cricket$& \begin{tikzpicture}
		    \node(a)[circle, fill, inner sep =1.5pt] at (0,0){};
            \node(b)[circle, fill, inner sep = 1.5pt] at(2,0){};
            \node(c)[circle, fill, inner sep = 1.5pt] at(0,1){};
            \node(d)[circle, fill, inner sep = 1.5pt] at(2,1){};
            \node(e)[circle, fill, inner sep = 1.5pt] at(1,0.5){};

            \draw (a)--(e)--(c);
            \draw (e)--(d)--(b)--(e);

            \node[fit=(a)(c)(e),dashed, draw, rectangle,rounded corners=10,inner sep=10pt] {};
            \node[fit=(b)(d),dashed, draw, rectangle,rounded corners=10,inner sep=10pt] {};
		\end{tikzpicture}&
    $C_5$ &\begin{tikzpicture}
		\node(a)[circle, fill, inner sep =1.5pt] at (0,0){};
            \node(b)[circle, fill, inner sep = 1.5pt] at(2,0){};
            \node(c)[circle, fill, inner sep = 1.5pt] at(0,1){};
            \node(d)[circle, fill, inner sep = 1.5pt] at(2,1){};
            \node(e)[circle, fill, inner sep = 1.5pt] at(1,0.5){};

            \draw (a)--(e);
            \draw (e)--(c)--(d);
            \draw (a)--(b);
            \draw (d)--(b);

            \node[fit=(a)(c)(e),dashed, draw, rectangle,rounded corners=10,inner sep=10pt] {};
            \node[fit=(b)(d),dashed, draw, rectangle,rounded corners=10,inner sep=10pt] {};
		\end{tikzpicture} \\

    \end{tabular}
    \caption{All of the decomposable graphs of order 5 with their non-trivial equilibrium partitions shown. Each has only one non-trivial equilibrium partition up to relabelling.}
    \label{n5decomposable}
\end{figure*}
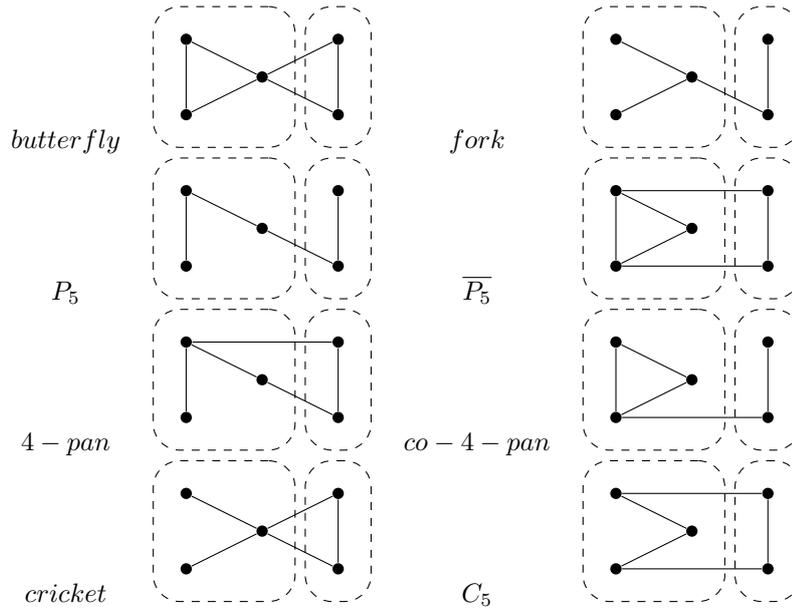

\subsection{n=6}\label{n6}
There are 112 non-isomorphic connected graphs on 6 vertices \cite{Cvetkovic1984}. Finding each of these partitions by hand is inefficient, so we relied solely on the the computer exhaustive search. The partitions are displayed visually in the online resource 1. The results of this search showed that 48 of the connected graphs on 6 vertices are indecomposable (Fig. A-1).
The remaining 64 connected graphs are decomposable. Of these, 43 have unique non-trivial equilibra. They are pictured in Figure A-2. 13 of them have exactly two non-trivial equilibrium partitions up to isomorphism (Fig. A-3), six have exactly three non-trivial equilibrium partitions, up to isomorohism (Fig. A-4), and the remaining two have exactly four non-trivial equilibrium partitions (Fig. A-5).

\subsection{n=7}
As inefficient as the visual examination method is for graphs of order 6, it is even less feasible for graphs of order 7. Therefore, we rely on the exhaustive search to gain insight into these partitions. For the obvious reason, we will not present a figure presenting all of them, either here or in the appendix,  but the code to construct the catalogue of partitions is available \href{https://zenodo.org/doi/10.5281/zenodo.11237078}{here}\cite{McAlisterStructuredCoord}.  Of the 853 graphs of order 7, 319 are indecomposable, 291 have a unique non trivial equilibrium partition. All others have multiple non-trivial equilibrium partitions including a single graph with ten distinct equilibrium partitions (Fig. \ref{graph445}). 
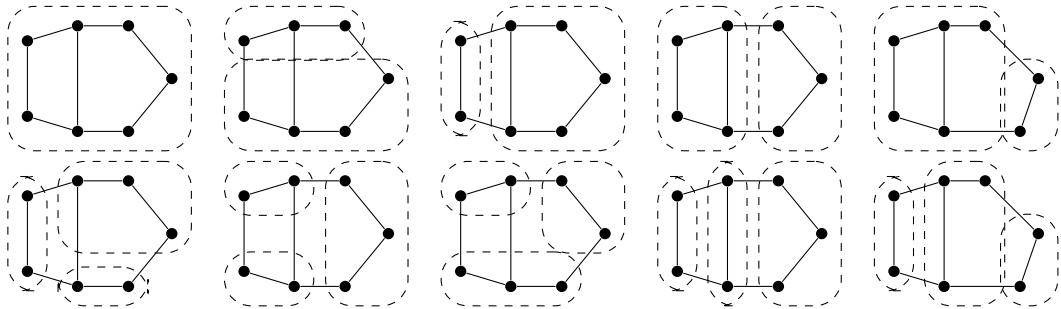
\begin{figure*}
    \centering
    \begin{tabular}{ccccc}
    \begin{tikzpicture}
            \node(a)[circle, fill, inner sep =1.5pt] at (0,0){};
            \node(b)[circle, fill, inner sep = 1.5pt] at(0.66,0.2){};
            \node(c)[circle, fill, inner sep = 1.5pt] at(1.33,0.2){};
            \node(d)[circle, fill, inner sep = 1.5pt] at(1.9,-0.5){};
            \node(e)[circle, fill, inner sep = 1.5pt] at(1.33,-1.2){};
            \node(f)[circle, fill, inner sep = 1.5pt] at(0.66,-1.2){};
            \node(g)[circle, fill, inner sep = 1.5pt] at(0,-1){};

            \draw (a)--(b)--(c)--(d)--(e)--(f)--(g)--(a);
            \draw (b)--(f);

            \node[fit=(a)(b)(c)(d)(e)(f)(g),dashed, draw, rectangle,rounded corners=10,inner sep=5pt] {};
         \end{tikzpicture}& 
    \begin{tikzpicture}
            \node(a)[circle, fill, inner sep =1.5pt] at (0,0){};
            \node(b)[circle, fill, inner sep = 1.5pt] at(0.66,0.2){};
            \node(c)[circle, fill, inner sep = 1.5pt] at(1.33,0.2){};
            \node(d)[circle, fill, inner sep = 1.5pt] at(1.9,-0.5){};
            \node(e)[circle, fill, inner sep = 1.5pt] at(1.33,-1.2){};
            \node(f)[circle, fill, inner sep = 1.5pt] at(0.66,-1.2){};
            \node(g)[circle, fill, inner sep = 1.5pt] at(0,-1){};

            \draw (a)--(b)--(c)--(d)--(e)--(f)--(g)--(a);
            \draw (b)--(f);

            \node[fit=(a)(b)(c),dashed, draw, rectangle,rounded corners=10,inner sep=5pt] {};
            \node[fit=(d)(e)(f)(g),dashed, draw, rectangle,rounded corners=10,inner sep=5pt] {};
            \end{tikzpicture}& 
    \begin{tikzpicture}
            \node(a)[circle, fill, inner sep =1.5pt] at (0,0){};
            \node(b)[circle, fill, inner sep = 1.5pt] at(0.66,0.2){};
            \node(c)[circle, fill, inner sep = 1.5pt] at(1.33,0.2){};
            \node(d)[circle, fill, inner sep = 1.5pt] at(1.9,-0.5){};
            \node(e)[circle, fill, inner sep = 1.5pt] at(1.33,-1.2){};
            \node(f)[circle, fill, inner sep = 1.5pt] at(0.66,-1.2){};
            \node(g)[circle, fill, inner sep = 1.5pt] at(0,-1){};

            \draw (a)--(b)--(c)--(d)--(e)--(f)--(g)--(a);
            \draw (b)--(f);

            \node[fit=(a)(g),dashed, draw, rectangle,rounded corners=10,inner sep=5pt] {};
            \node[fit=(b)(c)(d)(e)(f),dashed, draw, rectangle,rounded corners=10,inner sep=5pt] {};
    \end{tikzpicture}&
    \begin{tikzpicture}
            \node(a)[circle, fill, inner sep =1.5pt] at (0,0){};
            \node(b)[circle, fill, inner sep = 1.5pt] at(0.66,0.2){};
            \node(c)[circle, fill, inner sep = 1.5pt] at(1.33,0.2){};
            \node(d)[circle, fill, inner sep = 1.5pt] at(1.9,-0.5){};
            \node(e)[circle, fill, inner sep = 1.5pt] at(1.33,-1.2){};
            \node(f)[circle, fill, inner sep = 1.5pt] at(0.66,-1.2){};
            \node(g)[circle, fill, inner sep = 1.5pt] at(0,-1){};

            \draw (a)--(b)--(c)--(d)--(e)--(f)--(g)--(a);
            \draw (b)--(f);

            \node[fit=(a)(g)(b)(f),dashed, draw, rectangle,rounded corners=10,inner sep=5pt] {};
            \node[fit=(c)(d)(e),dashed, draw, rectangle,rounded corners=10,inner sep=5pt] {};
            \end{tikzpicture}&    \begin{tikzpicture}
            \node(a)[circle, fill, inner sep =1.5pt] at (0,0){};
            \node(b)[circle, fill, inner sep = 1.5pt] at(0.66,0.2){};
            \node(c)[circle, fill, inner sep = 1.5pt] at(1.2,0.2){};
            \node(d)[circle, fill, inner sep = 1.5pt] at(1.9,-0.5){};
            \node(e)[circle, fill, inner sep = 1.5pt] at(1.66,-1.2){};
            \node(f)[circle, fill, inner sep = 1.5pt] at(0.66,-1.2){};
            \node(g)[circle, fill, inner sep = 1.5pt] at(0,-1){};

            \draw (a)--(b)--(c)--(d)--(e)--(f)--(g)--(a);
            \draw (b)--(f);

            \node[fit=(a)(g)(b)(f)(c),dashed, draw, rectangle,rounded corners=10,inner sep=5pt] {};
            \node[fit=(d)(e),dashed, draw, rectangle,rounded corners=10,inner sep=5pt] {};
            \end{tikzpicture}\\

    \begin{tikzpicture}
            \node(a)[circle, fill, inner sep =1.5pt] at (0,0){};
            \node(b)[circle, fill, inner sep = 1.5pt] at(0.66,0.2){};
            \node(c)[circle, fill, inner sep = 1.5pt] at(1.33,0.2){};
            \node(d)[circle, fill, inner sep = 1.5pt] at(1.9,-0.5){};
            \node(e)[circle, fill, inner sep = 1.5pt] at(1.33,-1.2){};
            \node(f)[circle, fill, inner sep = 1.5pt] at(0.66,-1.2){};
            \node(g)[circle, fill, inner sep = 1.5pt] at(0,-1){};

            \draw (a)--(b)--(c)--(d)--(e)--(f)--(g)--(a);
            \draw (b)--(f);

            \node[fit=(a)(g),dashed, draw, rectangle,rounded corners=10,inner sep=5pt] {};
            \node[fit=(b)(c)(d),dashed, draw, rectangle,rounded corners=10,inner sep=5pt] {};
            \node[fit=(f)(e),dashed, draw, rectangle,rounded corners=10,inner sep=5pt] {};
            \end{tikzpicture}&
    \begin{tikzpicture}
            \node(a)[circle, fill, inner sep =1.5pt] at (0,0){};
            \node(b)[circle, fill, inner sep = 1.5pt] at(0.66,0.2){};
            \node(c)[circle, fill, inner sep = 1.5pt] at(1.33,0.2){};
            \node(d)[circle, fill, inner sep = 1.5pt] at(1.9,-0.5){};
            \node(e)[circle, fill, inner sep = 1.5pt] at(1.33,-1.2){};
            \node(f)[circle, fill, inner sep = 1.5pt] at(0.66,-1.2){};
            \node(g)[circle, fill, inner sep = 1.5pt] at(0,-1){};

            \draw (a)--(b)--(c)--(d)--(e)--(f)--(g)--(a);
            \draw (b)--(f);

            \node[fit=(a)(b),dashed, draw, rectangle,rounded corners=10,inner sep=5pt] {};
            \node[fit=(g)(f),dashed, draw, rectangle,rounded corners=10,inner sep=5pt] {};
            \node[fit=(c)(d)(e),dashed, draw, rectangle,rounded corners=10,inner sep=5pt] {};
            \end{tikzpicture}&
    \begin{tikzpicture}
            \node(a)[circle, fill, inner sep =1.5pt] at (0,0){};
            \node(b)[circle, fill, inner sep = 1.5pt] at(0.66,0.2){};
            \node(c)[circle, fill, inner sep = 1.5pt] at(1.33,0.2){};
            \node(d)[circle, fill, inner sep = 1.5pt] at(1.9,-0.5){};
            \node(e)[circle, fill, inner sep = 1.5pt] at(1.33,-1.2){};
            \node(f)[circle, fill, inner sep = 1.5pt] at(0.66,-1.2){};
            \node(g)[circle, fill, inner sep = 1.5pt] at(0,-1){};

            \draw (a)--(b)--(c)--(d)--(e)--(f)--(g)--(a);
            \draw (b)--(f);

            \node[fit=(a)(b),dashed, draw, rectangle,rounded corners=10,inner sep=5pt] {};
            \node[fit=(g)(f)(e),dashed, draw, rectangle,rounded corners=10,inner sep=5pt] {};
            \node[fit=(c)(d),dashed, draw, rectangle,rounded corners=10,inner sep=5pt] {};
            \end{tikzpicture}&
    \begin{tikzpicture}
            \node(a)[circle, fill, inner sep =1.5pt] at (0,0){};
            \node(b)[circle, fill, inner sep = 1.5pt] at(0.66,0.2){};
            \node(c)[circle, fill, inner sep = 1.5pt] at(1.33,0.2){};
            \node(d)[circle, fill, inner sep = 1.5pt] at(1.9,-0.5){};
            \node(e)[circle, fill, inner sep = 1.5pt] at(1.33,-1.2){};
            \node(f)[circle, fill, inner sep = 1.5pt] at(0.66,-1.2){};
            \node(g)[circle, fill, inner sep = 1.5pt] at(0,-1){};

            \draw (a)--(b)--(c)--(d)--(e)--(f)--(g)--(a);
            \draw (b)--(f);

            \node[fit=(a)(g),dashed, draw, rectangle,rounded corners=10,inner sep=5pt] {};
            \node[fit=(b)(f),dashed, draw, rectangle,rounded corners=10,inner sep=5pt] {};
            \node[fit=(c)(d)(e),dashed, draw, rectangle,rounded corners=10,inner sep=5pt] {};
            \end{tikzpicture}&
    \begin{tikzpicture}
            \node(a)[circle, fill, inner sep =1.5pt] at (0,0){};
            \node(b)[circle, fill, inner sep = 1.5pt] at(0.66,0.2){};
            \node(c)[circle, fill, inner sep = 1.5pt] at(1.2,0.2){};
            \node(d)[circle, fill, inner sep = 1.5pt] at(1.9,-0.5){};
            \node(e)[circle, fill, inner sep = 1.5pt] at(1.66,-1.2){};
            \node(f)[circle, fill, inner sep = 1.5pt] at(0.66,-1.2){};
            \node(g)[circle, fill, inner sep = 1.5pt] at(0,-1){};

            \draw (a)--(b)--(c)--(d)--(e)--(f)--(g)--(a);
            \draw (b)--(f);

            \node[fit=(a)(g),dashed, draw, rectangle,rounded corners=10,inner sep=5pt] {};
            \node[fit=(b)(f)(c),dashed, draw, rectangle,rounded corners=10,inner sep=5pt] {};
            \node[fit=(d)(e),dashed, draw, rectangle,rounded corners=10,inner sep=5pt] {};
            \end{tikzpicture}     
    \end{tabular}
    \caption{The ten distinct equilibrium partitions of the graph $X_{38}$, which is graph \#445 in the graph atlas. This is the only graph among those catalogued which admits ten distinct equilibrium partitions.}
    \label{graph445}
\end{figure*}

\subsection{Catalogue Summary}
In total there are 996 connected graphs on at most 7 vertices and so there are 996 trivial equilibrium partitions. There are 597 graphs which admit non-trivial equilibrium partitions. The majority of these (350) have exactly 1 non-trivial equilibrium, 141 have two non-trivial equilibria, 62 have three non-trivial equilibra, and 55 have four or more non-trivial equilibria. This data is separated by graph size in Table \ref{tab:CountPartitionNumber}. 

Examining all 2083 partitions, we find that 996 of them are trivial (clearly because there are 996 different graphs in the catalogue), 897 have 2 parts, and 190 have three parts. We know this is the upper limit for graphs this small by lemma \ref{clusters of size 1}. This data is separated by graph size in table \ref{NumberClusters}



\begin{table*}[]
    \centering
    \begin{tabular}{c|cccccccccc}
         &\multicolumn{10}{|c}{number of graphs with $n$ partitions}  \\
         \hline 
         Graph size&1&2&3&4&5&6&7&8&9&10\\
         \hline 
         1&1&0&0&0&0&0&0&0&0&0\\
         2&1&0&0&0&0&0&0&0&0&0\\
         3&2&0&0&0&0&0&0&0&0&0\\
         4&4&2&0&0&0&0&0&0&0&0\\
         5&13&8&0&0&0&0&0&0&0&0\\
         6&48&43&13&6&2&0&0&0&0&0\\
         7&319&297&128&56&25&15&8&3&1&1\\
         \hline
         Total&399&350&141&62&27&15&8&3&1&1
    \end{tabular}
    \caption{A table showing the number of connected graphs which admit $n$ different equilibrium partitions for $n$ from 1 to 10, up to isomorphism. Among graphs of size less or equal to seven, there are no graphs which admit more than 10 different partitions.}
    \label{tab:CountPartitionNumber}
\end{table*}

\begin{table*}[h!]
    \centering
    \begin{tabular}{c|cccccc}
         & \multicolumn{6}{|c}{number of partitions with $n$ parts} \\
         \hline 
         Graph Size&\multicolumn{2}{|c}{1}&\multicolumn{2}{c}{2}&\multicolumn{2}{c}{3}\\
         \hline 
        1&1&100\%&0&0\%&0&0\%\\
        2&1&100\%&0&0\%&0&0\%\\
        3&2&100\%&0&0\%&0&0\%\\
        4&6&75\%&2&25\%&0&0\%\\
        5&21&72\%&8&27\%&0&0\%\\
        6&112&54\%&79&38\%&16&8\%\\
        7&853&46\%&808&44\%&174&9\%\\
        \hline
        Total& 996&48\%&897&43\%&190&9\%
    \end{tabular}
    \caption{A table showing the number of distinct partitions with $n$ clusters for each graph size from 1 to 6 vertices. }
    \label{NumberClusters}
\end{table*}

A crucial observation is that indecomposability becomes less common as graph order increases. We cannot yet rigoruosly extend this observation past $n=7$ but, through visual examination of the catalogue, we can observe that indecomposability depends greatly the particular topology of the graph and is clearly less common when the edge number of a graph decreases. We will revisit this observation in Section \ref{Discussion} to formulate some conjectures about indecomposability in infinitely large $\Gamma_{n,N(n)}$ graphs. 

\section{Larger Simulation}\label{Simulation}
Having a complete catalogue of equilibria is helpful for small graphs where the equilibria are easy to draw. However, it becomes unhelpful and unreasonably time intensive to make a catalogue for larger graphs. Instead, we turn to simulation in order to find equilibria. Because of the growth of Bell's numbers, we no longer simply check all partitions in search of equilibrium partitions; instead, we search for equilibrium partitions by considering the initial value problem \eqref{dsys}. By producing random initial conditions, we can run the myopic best response process until the strategy profile ends in an equilibrium. Using this strategy, there are several things to consider: The first is that solutions may terminate in cyclical behavior where $u(t+n)=u(t)$ for all $t>T$ and $n>1$. These states are interesting and worth studying, but they are not equilibrium partitions. Additionally, we must contend with the basins of stability of different equilibria. It may be that case that a graph admits several equilibrium partitions, but the basin of stability for one equilibria is so large that it is unlikely for randomly generated initial data to result in a solution which tends towards a different equilibrium. 

In order to discuss the basin of stability for a graph, we introduce the notation $\phi:\mathcal{P}(\mathcal{Q}_G)\rightarrow [0,1]$, where $\mathcal{Q_G}$ is the set of vertex partitions of $G$, $\mathcal{P}(\mathcal{Q}_G)$ is the power set of $\mathcal{Q}_G$, and $\phi((Q_i)_{i=1}^n)$ is the probability that, given a random initial strategy profile, the IVP will converge to the sequence $(Q_i)_{i=1}^n$. In the case that we are interested in the convergence to a partition $Q$, we write $\phi(Q)$. This means that $\sum_{(Q_i)\in\mathcal{P}(\mathcal{Q})}\phi((Q_i))\leq1$ (with equality when a graph $G$ does not admit any non-convergent solutions). Moreover, if $Q^0$ is the trivial partition (which corresponds to the consensus equilibrium) and $G$ is finite, then $\phi(Q_G^0)=1\implies $ $G$ is indecomposable. This relationship is not true in the case that $G$ is infinite because we must contend with events of measure zero. Note that the converse is not true even when $G$ is finite because, by definition, a graph my be indecomposable but still admit cycling (e.g, $K_{1,n}$).

\subsection{Basins of Stability for consensus equilibria}
In the first of two simulations, we attempt to understand the size of the basin of stability of the trivial equilibrium, $\phi(Q^0)$, in many different graphs. Using the Erd\H{o}s-R\'enyi (ER) algorithm to generate 10,000 random graphs with various sizes and ednge densities, we produced solutions from 500 randomly selected initial strategy profiles, and record the limiting behavior of the solution. In total, approximately 5 million solutions were attempted and their limits were measured. It is important to note that the basins of stability in this sense are not exactly partitions of the state space but rather partitions of unity over the state space. Because of the stochastic tie breaking, basins of stability overlap. Therefore, when we measure the relative size of a basin of stability, it is not just the cardinality of a subset of the state space divided by the size of the state space itself. Instead, it is a weighted average over the entire state space of the probability that that state will evolve to the equilibrium in questions.   
\begin{figure}[h!]
    \centering
    \includegraphics[width = \linewidth]{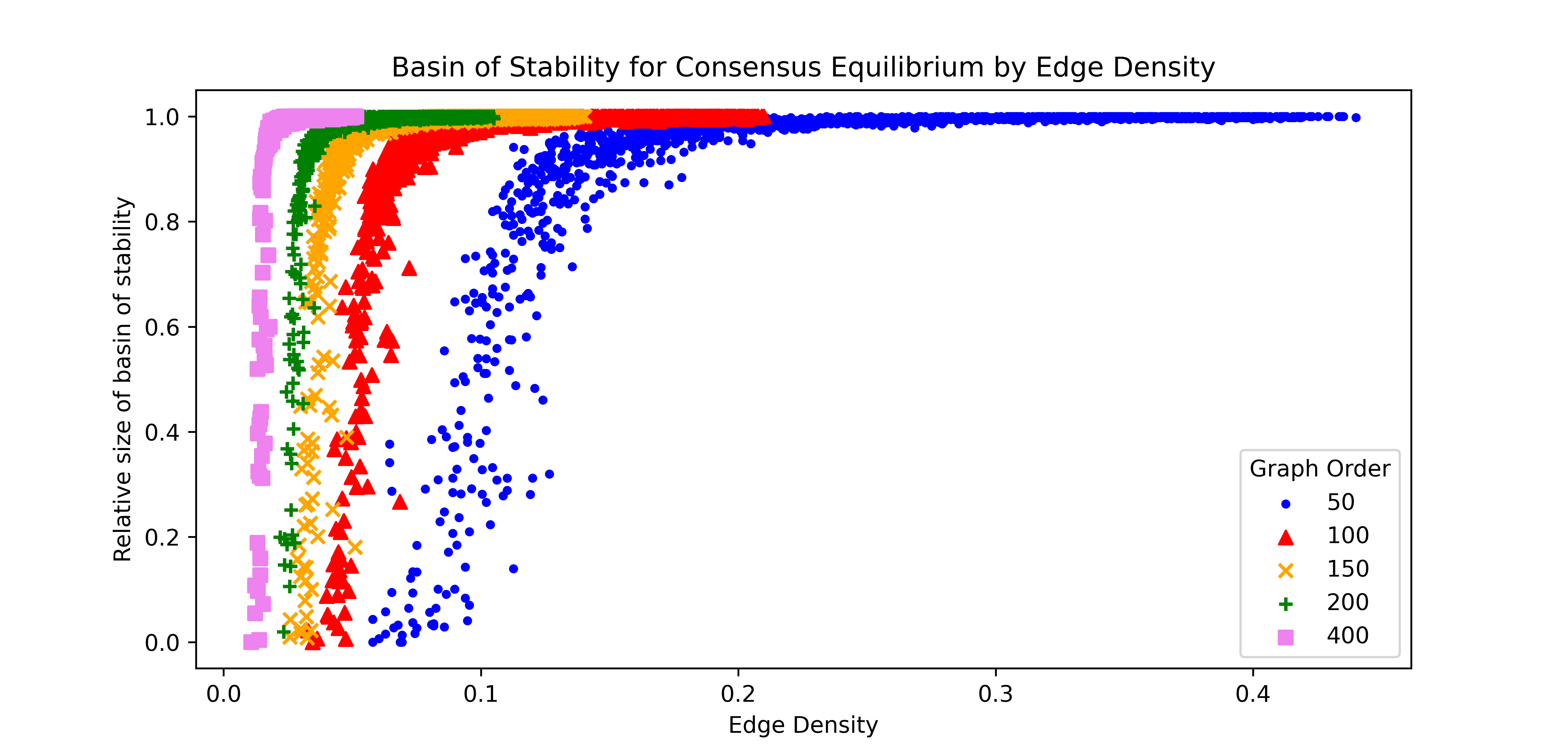}
    \includegraphics[width =\linewidth]{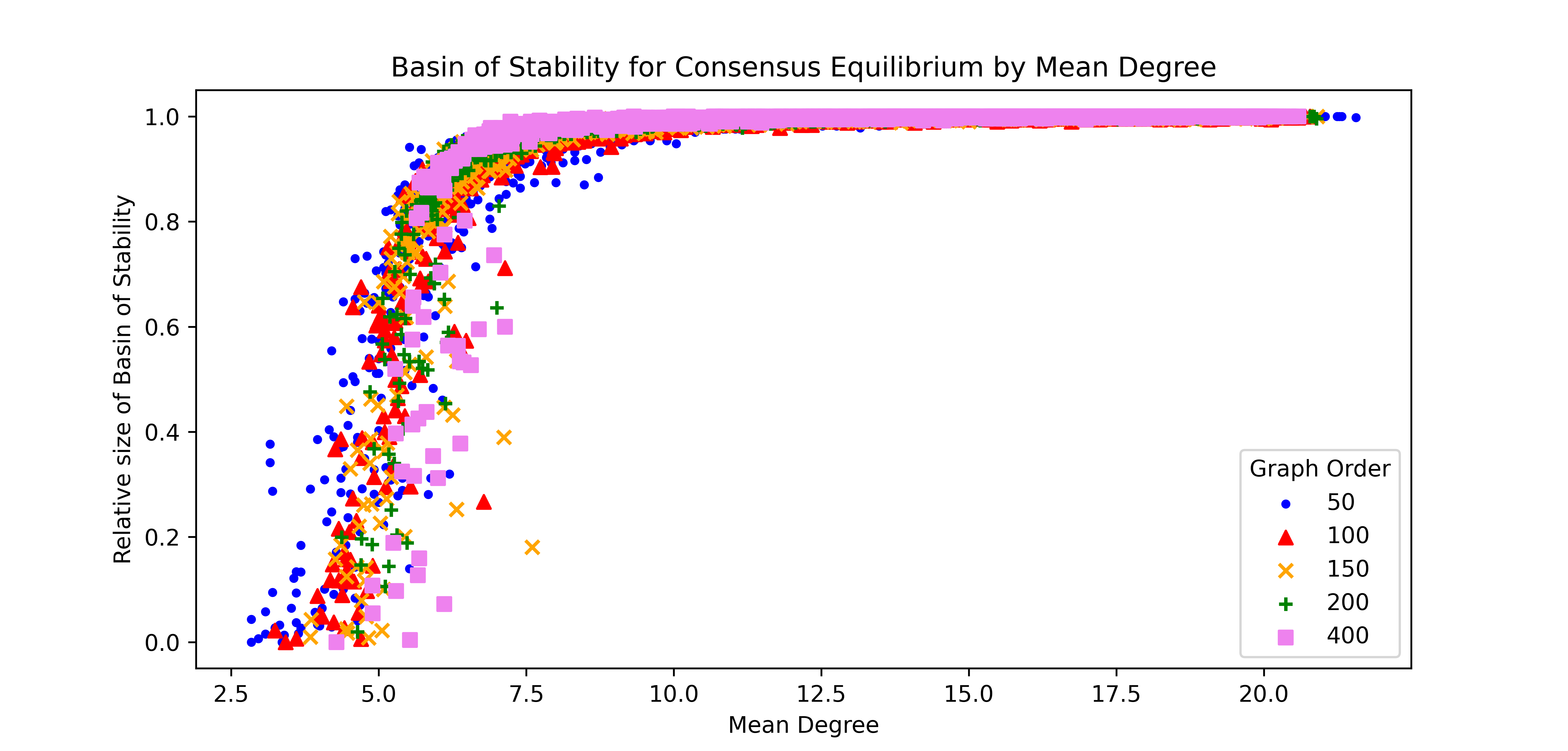}
    \caption{\textbf{Top:} A scatter plot showing every random graph with edge density on the x-axis and relative proportion of the time the consensus equilibrium was reached on the y-axis. \textbf{Bottom:} A similar scatter plot with mean degree on the x-axis instead of edge density. Both plots show the clear trend that the probability of arriving at the consensus equilibrium increases as ``connectivity" (in whatever sense we consider) increases.}
    \label{BasinSimSummary}
\end{figure}

The main result of this simulation was the apparent increase in  $\phi(Q^0)$ with the increase in connectivity of the graph (Fig.\ref{BasinSimSummary}). Something equally interesting but more subtle is that when recorded against edge density, the data is highly stratified by graph order. If instead the data is measured against mean degree, this stratification is no longer apparent. This is consistent with the fact that the Nash equilibrium is a highly local property. The property of being in equilibrium depends only on those vertices with which a vertex shares an edge and does not depend directly on the order of the graph itself.

Included in this simulation are only those ER graphs which were connected. We can say certainly that the probability that a disconnected graph arrives at the consensus equilibrium is bounded above by $\frac{1}{2}$. This is because, even if each connected component will surely evolve to a consensus equilibrium of its own, the likelihood of each connected component evolving to the same consensus equilibrium is less than 1. Indeed, if there are $r$ connected components in the graph $G$, even if every connected component will certainly evolve to a consensus, if the initial conditions are given uniformly randomly, the probability of converging to a global consensus is only $\phi(Q^0)=\frac{1}{r^{c-1}}$, where $c$ is the number of strategies present in the initial condition. In every non-trivial case, where $r>1$, this probability is less than $\frac{1}{2}$.

$\phi(Q^0)$ does seem to be related to the probability of being connected in the ER graph given by the recursive expression
\begin{equation}
\mathbb{P}(\Gamma_{n,p}\in C)= 1-\sum_{i=1}^{n-1}\mathbb{P}(\Gamma_{i,p}\in C){n-1\choose i-1}(1-p)^{i(n-i)}    
\end{equation}
where $\Gamma_{n,p}$ is the random ER graph with parameters $n$ and $p$ ($p$ is the probability that any two vertices share an edge), and $C$ is the set of all graphs which are connected.  When $\Gamma_{n,p}$ with the same parameters is very likely to be connected, then we observe that $\phi(Q^0)$ is close to 1 (Fig \ref{BasinSimConnectedness}). We consider this relationship further in section \ref{Discussion}.
\begin{figure}[h!]
    \centering
    \includegraphics[width = \linewidth]{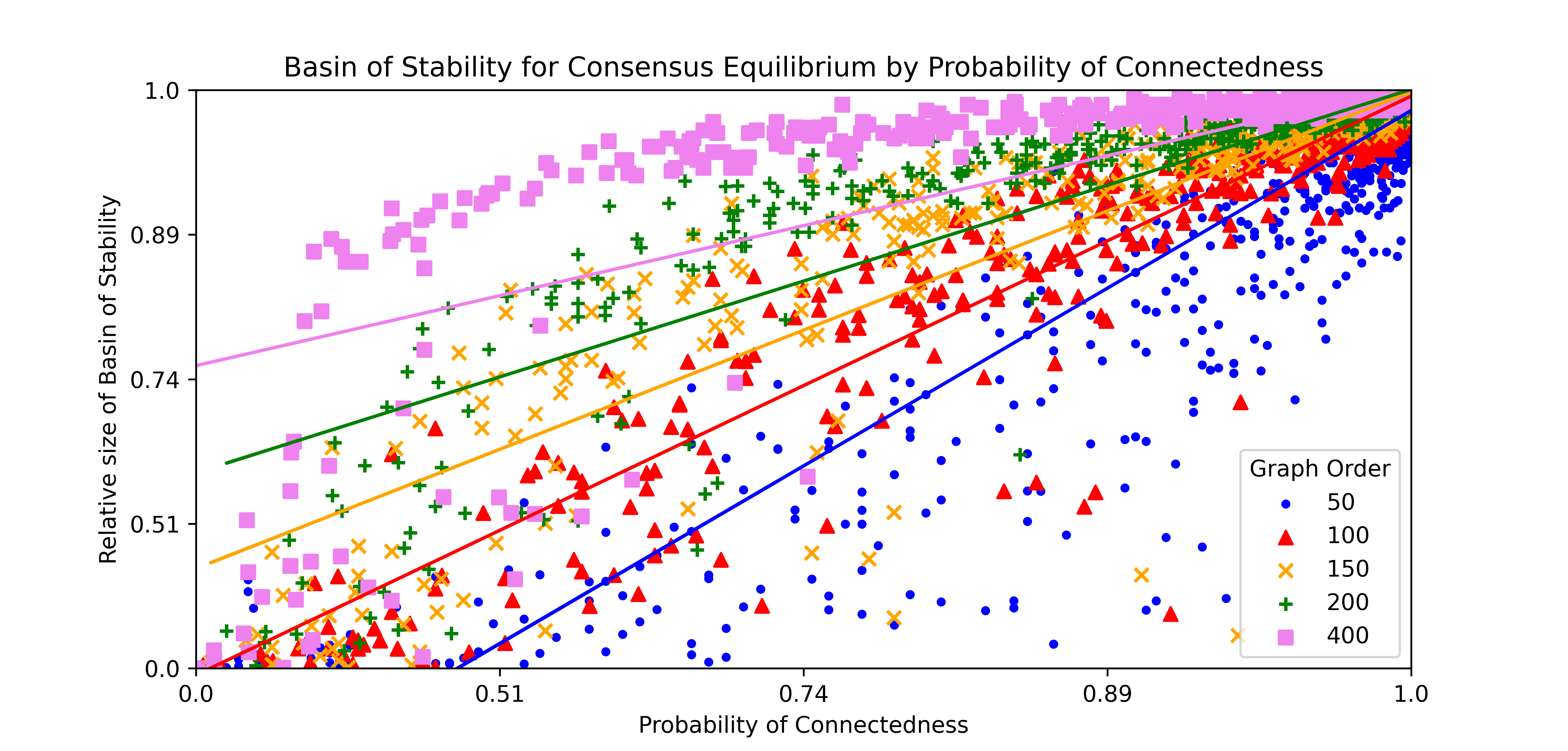}
    \caption{A scatter plot showing the relative size of the basin of stability as it is related to the probability of connectedness, shown on exponentially transformed axes. Trend lines are included for each series ($n=50,100,150,200,400$ from bottom to top), not because we hypothesize that this relationship is linear, but in order to demonstrate that this relationship is supported even though certain parts of the data are more dense than others. }
    \label{BasinSimConnectedness}
\end{figure}

\subsection{Broader Simulations}
The basin of stability simulation tells us about how the the graphical structure can promote or suppress consensus equilibrium, but it tells us less about particular solutions to the IVP and their limits. In order to understand the limit of the IVP a little better, we generated 1 million Erd\H{o}s-R\'enyi random graphs, discarded the disconnected graphs, and found solutions to the IVP with randomly generated initial data to give us a (lower resolution) picture of a wider range of graphs. In particular, this simulation demonstrates better how graph size effects a broader range of limiting behaviors.

\begin{figure}
    \centering
    \includegraphics[width = \linewidth]{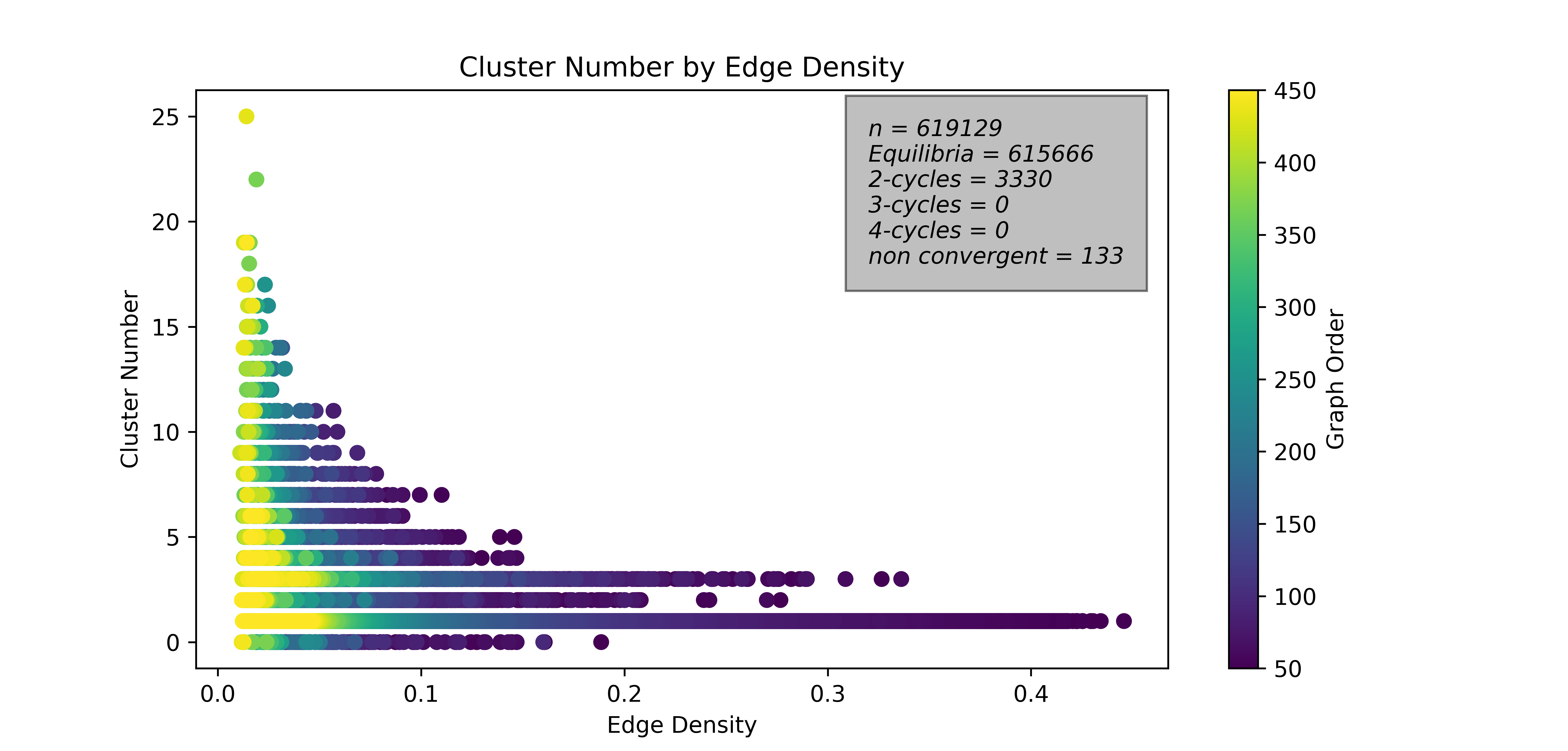}
    \caption{A scatter plot of cluster number (which is the number of distinct strategies present at equilibrium) as it relates to the edge density of the graph. This plot also shows that of the 619,129 IVPs that were simulated, 614,666 of them ended in an equilibrium, 3330 ended in a 2-cycle, and 113 did not converge by the end of the simulation. There were no 3-cycles or 4-cycles observed across either of the simulations. Any point with a cluster number of 0 corresponds to a solution which did not converge. }
    \label{BroadSimMainFig}
\end{figure}

Cluster number is the number is strategies present at equilibrium, or in the language of partitions, it is simply the number of parts to the equilibrium partition. When plotted against edge density of the graph by which the equilibrium was admitted, we see a story consistent with that of the basin of stability simulation (Fig \ref{BroadSimMainFig}). When edge density increases, the number of strategies present at equilibrium tends towards 1 (which corresponds to the consensus equilibrium). Additionally, this simulation shows a little more about how this interaction looks on the \textit{graph order}$\times$ \textit{mean degree} plane (which is a transformation of the parameter space for the ER graph random variable $\Gamma_{n,p}$) (Fig. \ref{BroadSimHeatMaps}).

\begin{figure}
    \centering
    \includegraphics[width=0.85\linewidth]{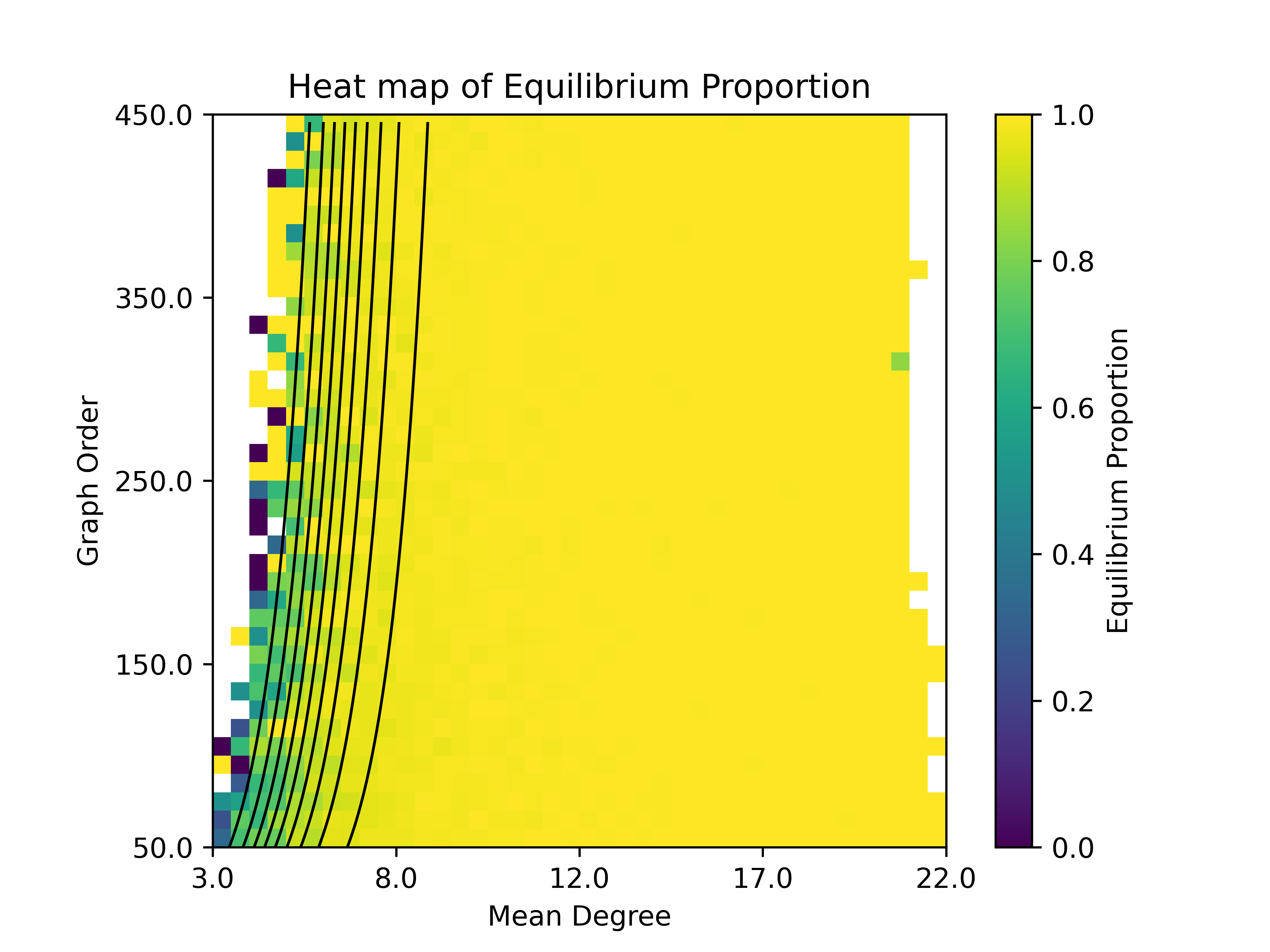}
    \includegraphics[width = 0.85\linewidth]{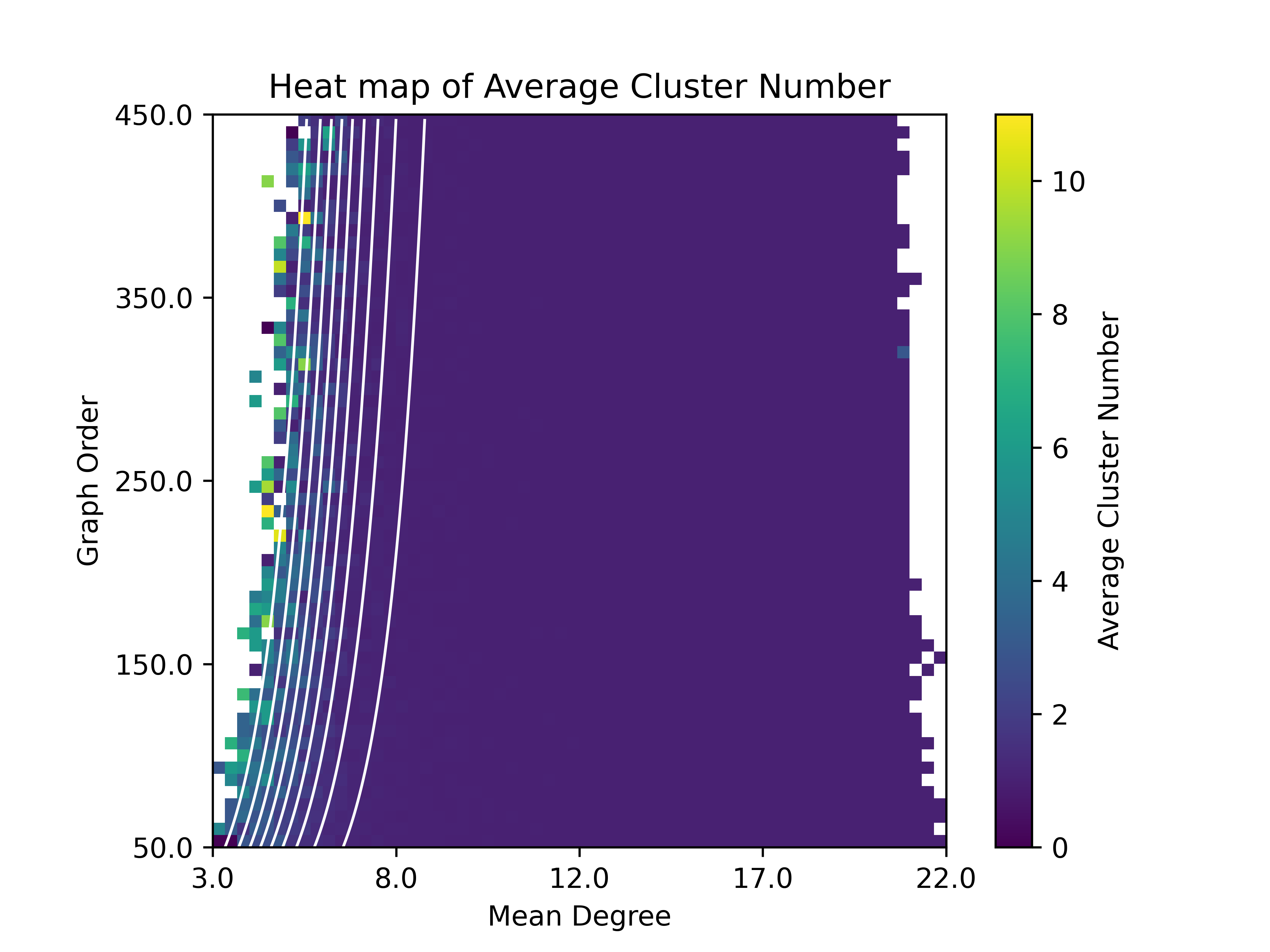}
    \caption{\textbf{Top:} a heat map showing the frequency of converging to an equilibrium for each combination of graph order and mean degree. Clearly, the probability of converging to an equilibrium approaches 1 as mean degree increases for any choice of graph order. \textbf{Bottom:} a heat map showing the average cluster number for each choice of graph order and mean degree. On both images are superimposed level curves of the asymptotic probability of connectedness ranging from $p=0.1$ to $p=0.9$}
    \label{BroadSimHeatMaps}
\end{figure}

When visualized on the on the \textit{graph order}$\times$ \textit{mean degree} plane, the non-trivial dynamics seem only to occur near the frontier where it becomes exceedingly unlikely for an ER graph to be connected. A foundational result to the study of random graphs is the asymptotic probability of connectedness \cite{Karonski1982}. For the purpose of visualizing the results, we have superimposed level sets of the approximated probability of of connectedness by way of the sharp threshold functions found in \cite{Erdos1959} and \cite{erdHos1960evolution}.  It is clear that these level sets divide the plane into regions with qualitatively different behavior. We examine the results of the simulation though the lens of connectedness probability further in section \ref{Discussion}.

The other notable result form this simulation is that, although 2 cycles were relatively common, three and four cycles were entirely absent. Among all the IVPs across both simulations there was never an $n$-cycle for $n>2$. (Although among the 133 solutions that did not converge to a recognized equilibrium before the simulation timed out, there was a common behavior which mimicked a 4 cycle, which is further discussed in section \ref{Discussion}.)

\section{Discussion}\label{Discussion}

The ultimate goal of this endeavor is not to have a statistical understanding of this system, but rather an analytical understanding. As commented on in section \ref{Background}, there are few general results about equilibria of this system or the behavior of solutions to the IVP. However, with the help of the the numerical results, we propose and discuss three conjectures here. The first is that when graphs grow infinitely large, the probability of an ER graph being indecomposable approaches zero so long as the number of edges is below some threshold function. The second is that as graphs grow infinitely large, $\phi(Q_0)$ converges to 1 so long as the number of edges is above some threshold function. These two conjectures together would imply the existence of a region wherein infinitely large graphs are almost surely decomposable but the non-trivial equilibria are almost never discoverable. The last is that $n$-cycles are impossible for $n>2$. 

These conjectures will help shape the way we understand this critical case of the coordination game and will, along with the results of the present numerical treatment, inform our understanding of the application areas of this system. For instance, if conjecture 2 is shown to be true, it would imply that, for the canonical example of a growing company selecting between operating systems, there is a critical amount of interconnectedness, above which the company will eventually converge to a consensus, almost surely, without intervention. If conjecture 3 is proven, it may give us further theoretical support to the idea that cyclical trends in fashion cannot be described only by individuals choosing to mimic one another \cite{Alberto2012}.

\subsection{Indecomposablity becomes asymptotically rare}
\begin{conjecture}
    There is a function $A(n)$ which is monotonic increasing with $\lim_{n\rightarrow \infty} A(n)=\infty$ so that, if $P_{id}(n,N(n))$ is the probability that $\Gamma_{n,N(n)}$ is indecomposable, 
    \begin{equation}
       \lim_{n\rightarrow \infty} P_{id}(n,N(n))=\begin{cases}
            0 &\lim_{n\rightarrow \infty}\frac{N(n)}{A(n)}=0\\
            1 &\lim_{n\rightarrow \infty}\frac{N(n)}{A(n)}=\infty
        \end{cases}
    \end{equation}
\end{conjecture}
From the catalogue we start to see that as $n$ grows, a smaller proportion of graphs of order $n$ are indecomposable. Clearly, we cannot draw meaningful asymptotic results from observations of $n\in \{1,2,...,7\}$, but we can support this conjecture with some facts about the system. First it is clear that any disconnected graph is decomposable. This means immediately that any lower threshold function (a la \cite{erdHos1960evolution}) for connectedness, like $A(n)=\frac{1}{2}n\log n$, is also trivially a lower threshold function for indecomposability. That is, if we consider random graphs $\Gamma_{n,N(n)}$ and let $n\rightarrow \infty$, the probability that $\Gamma_{n,N(n)}$ is indecomposable goes to zero when $\lim \frac{N(n)}{A(n)}= 0$.

We propose that a threshold function exists and is much stronger than the trivial lower threshold function. It is not true that every graph of size $n$ is indecomposable when $n\rightarrow \infty$. For instance, $K_n$ and $K_n-e$ (for $n>3$) are always indecomposable, as is $K_{n,m}$ when $n$ and $m$ are coprime (recall theorems \ref{KnTheorem} and \ref{KnnTheorem}). Obviously then, if $N(n)\geq {n \choose 2}-1$, then $\Gamma_{n,N(n)}$ is indecomposable. Certainly we know that if a threshold function exists it will satisfy $\lim\frac{{n\choose 2}-1}{A(n)}= \infty$ and $\lim\frac{(\frac{1}{2})n\log n}{A(n)}= 0$. The path towards finding such a threshold function, we believe, will involve the study of the separating sets of critically large subgraphs. 

A proof of this conjecture would provide some idea about the behavior of very large networks and community formation. In the same way that \cite{McDiramid2020} discusses modularity of ER networks towards understanding the community properties of large graphs, we may be able to discuss the properties of community formation in very large graphs as this system provides a similar, but more local, conception of the community structure in a network. A separate goal of this work is to be able to algorithmically find equilibrium partitions of finite graphs to describe community formation, just as \cite{Newman2006} and \cite{Clauset2004} do with modularity, but this conjecture would inform our understanding of the system in a way that algorithms cannot.

\subsection{Discovery of non-trivial equilibria becomes asymptotically rare}
For this conjecture recall that the consensus strategy profiles make up the equivalence class associated with the trivial partition $Q^0$.
\begin{conjecture}
    There exists a function $B(n)$ which is monotonic with $\lim_{n\rightarrow \infty} B(n)=\infty$ so that, if $\phi(Q_{\Gamma_{n,N(n)}}^0)$ is the probability of a random initial strategy profile evolving into a consensus strategy profile, under the dynamics in $\eqref{dsys}$ in the random graph $\Gamma_{n,N(n)}$, then
    \begin{equation}
        \lim_{n\rightarrow \infty}\phi(Q_{\Gamma_{n,N(n)}}^0)=\begin{cases}
            0&\lim_{n\rightarrow \infty}\frac{N(n)}{B(n)}=0\\
            1&\lim_{n\rightarrow\infty}\frac{N(n)}{B(n)}=\infty
        \end{cases}
    \end{equation}
\end{conjecture}
The results of second simulation suggest that $B(n)=\frac{1}{2}n\log(n)$ (which, again, is a threshold function for connectedness \cite{erdHos1960evolution}) could be such a threshold function. The results shown in Figures \ref{BroadSimMainFig} and \ref{BroadSimHeatMaps} show that as graph order increases, once edge density passes a threshold close to the threshold for connectedness, the probability of reaching the consensus equilibrium (or equivalently having a cluster number of 1) approaches 1.

Again, as discussed in section \ref{Simulation}, when the graph is disconnected the relative size of the consensus basin of stability grows as $\frac{1}{r^{c-1}}$, where $r$ is the number of components and $c\geq 2$. By theorem 6 of \cite{erdHos1960evolution}, we see that if $N(n)\sim cn$ for any $c<1$ the expected number of components will tend to infinity as $n\rightarrow \infty$ and so the relative size of the consensus basin of stability will go to zero. This implies that $B(n)\sim cn$ is certainly  a lower threshold function for this property. We conjecture the existence of a threshold function which is a stronger lower threshold function than $cn$. Progress on this question will include examining stability of equilibria. 

It is clear that, should a threshold function exist for both indecomposability ($A(n)$) and for the basin of stability for the consensus equilibrium being 1 ($B(n)$), then $A(n)\geq B(n)$. Moreover, we conjecture that, should these two exist, this inequality will be strict. That is, there will be a region wherein, asyptotically, almost every $\Gamma_{n,N(n)}$ has non-trivial equilibria but they are discoverable with probability 0. This conjecture further gives insight into the nature of community formation in large networks.  Specifically, it shows that through these dynamics of community formation, the non-trivial equilibria are only finite-size effects which do not appear in the infinite system.  
    
\subsection{Impossibility of n-cycles ($n>2$)}
\begin{conjecture}
    There is no graph $G$ such that there exists $n>2$ strategy profiles $\{\uu(1),\uu(2),...,\uu(n)\}$ such that $|\argmax_{c\in C}\{w(v,c|\uu(t))\}|=1$ for all $t=1,2,...,n$ and 
    \begin{equation}
        \begin{split}
            \uu_v(2)&=\argmax_{c\in C}\{w(v,c|\uu(1))\}\\
            \uu_v(3)&=\argmax_{c\in C}\{w(v,c|\uu(2))\}\\
            &\vdots\\
            \uu_v(1)&=\argmax_{c\in C}\{w(v,c|\uu(n))\}
        \end{split}
    \end{equation}
\end{conjecture}

Throughout this entire project, throughout every attempt to run each simulation, neither a true 3-cycle nor a true 4-cycle was ever discovered. This suggests that, at least, these kind of limits are exceedingly unstable, but we further conjecture that n-cycles with $n>2$ exist with probability 0.

Frequently we observe fixed points and 2-cycles which, in the state space, behave as in figure \ref{Statespace12}.  These exist deterministically and are uncomplicated to understand. Indeed, $K_2$ will exhibit a 2-cycle whenever the initial condition is non-trivial. It is not yet clear why $n>2$-cycles do not appear as readily. If such a cycle exists stochasticly it will continue to cycle for all time with probability zero, so to prove our conjecture we will need to show that deterministic $n>2$ cycles are impossible.

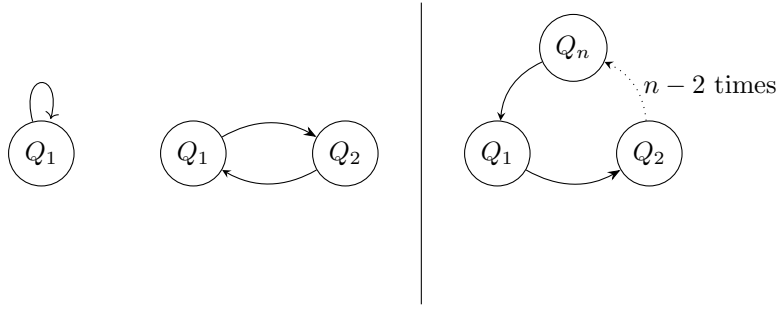
\begin{figure}
    \centering
    \begin{tikzpicture}
        \node (A) [draw = black, circle] at (0,0){$Q_1$};

        \draw  (A)edge [loop above](A);

        \node (B) [draw = black,circle] at (2,0){$Q_1$};
        \node (C) [draw = black,circle] at (4,0){$Q_2$};

        \draw (B) edge [bend left, arrows = -Stealth] (C);
        \draw (C) edge [bend left, arrows = -stealth] (B);

        \draw (5,2)--(5,-2);

        \node (D) [draw = black,circle] at (6,0){$Q_1$};
        \node (E) [draw = black,circle] at (8,0){$Q_2$};
        \node (F) [draw =black,circle] at  (7,1.4){$Q_n$};
        \draw (D) edge [bend right, arrows = -Stealth ] (E);
        \draw (E) edge [bend right, arrows = -stealth, dotted] node[right]{$n-2$ times} (F);
        \draw (F) edge [bend right, arrows = -stealth] (D);
    \end{tikzpicture}
    \caption{$Q_i$ are elements of the state space. On the left, 1-cycles (which is an equilibrium) and a 2-cycles are frequently observed and behave deterministically. On the right an $n$-cycle with $n>2$ is never observed deterministically and stochastically will exist with probability 0.}
    \label{Statespace12}
\end{figure}
In the simulations, the conditional statements which check for an $n$-cycle have a false positive rate of $\leq 2^{n-T}$, where $T$ is the tolerance (which was set to 34 so that the false positive rate, even for $4-$ cycles, is $ <10^{-9}$), but during early iterations of the simulation (when the filter had a false positive rate of only $\leq 2^{-\lfloor\frac{T-n}{n}\rfloor}$) we would see rare false positive 4-cycle which could pass through the filter with probability $2^{-7}$ with the behavior demonstrated in figure  \ref{Statespace4}

\begin{figure}[h!]
    \centering
    \begin{tikzpicture}
        \node (D) [draw = black,circle] at (0,0){$Q_1$};
        \node (E) [draw = black,circle] at (-1.4,2){$Q_2$};
        \node (F) [draw =black,circle] at  (1.4,2){$Q_3$};
        \draw (D) edge [bend left, arrows = -Stealth] node[left] {$p=0.5$} (E);
        \draw (E) edge [bend left, arrows = -Stealth] (D);
        \draw (D) edge [bend right, arrows = -stealth] node [right] {$p=0.5$}(F);
        \draw (F) edge [bend right, arrows = -stealth] (D);
    \end{tikzpicture}
    \caption{A pseudo 4-cycle which ``starts" at $Q_1$ then with probability $0.5$ proceeds to either $Q_2$ or $Q_3$. At either of these state, it certainly returns to $Q_1$. This behavior can appear like at 4 cycle with non-zero probability but will continue to behave as a 4-cycle with probability zero as time extends to infinity. }
    \label{Statespace4}
\end{figure}
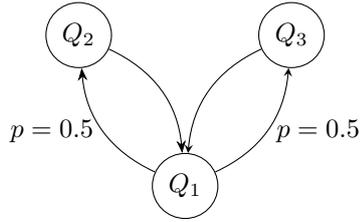

These pseudo 4 cycles give some suggestions as to why true deterministic $n$ cycles with $n>2$ do not appear in our observation. In order for a state to be reached deterministically, every best response set must be of size 1. The resulting strict inequalities present the main limitation. Aside from being a remarkable feature of this dynamical system, the impossibility of large cycles, again, offers new insights into the process of community formation when is driven by local interactions alone. 

\subsection{General Discussion}
Here we have introduced the structured coordination game with neutral options and we have seen that the dynamics are categorically different from the dynamics of structured coordination game when one strategy may dominate the other in terms of payoff or risk avoidance and so cannot be studied in the same way. When noise enters the system through stochastic tie breaking rather than through random mutation, as in \cite{Kandori1993,Ellison1993,Weidenholzer}, we cannot use use the radius and coradius of absorbing sets to find long run equilibrium as in \cite{Ellison1993}. However, we conjecture that asymptotically many of the same results (adapted to ignore payoff/risk dominance) remain in place. That is, in general, we can expect to converge to a consensus equilibrium even if other equilibria are present. 

The main result is that higher connectivity of the network promotes convergence to the consensus equilibrium. This is not an altogether surprising result. It is, however, critically different from the conclusions of \cite{Buskens2016} that network density does not have an effect on the likelihood of behavior to converge to a homogeneous state. These differences can be explained both by the differences in graph order being observed in the studies and by the fact that, in the non-neutral case, differences in payoff mean that converging to a non-consensus equilibrium is already rare, even at low edge densities. The graphs considered herein were, in general, nearly amodular with high probability \cite{Campigotto2013}, however, near the parameter space where $\Gamma_{n,p}$ is not likely to be connected, the expected modularity increases with high probability to above 0.4 \cite{McDiramid2020}. Thus, as in \cite{Buskens2016}, we observe that higher modularity is associated with non-trivial equilibria (which they call persistence of heterogeneous behavior).  

This simulation approach has helped lay the groundwork for an exciting body of analytical work to come. With a complete catalogue of equilibria for small graphs, as well as extensive simulated results in Erd\H{o}s-R\'enyi graphs, we now have robust evidence to support our three main conjectures: 1) the existence of a threshold function for indecomposability which is close to ${n\choose 2}-cn$, 2) the existence of a threshold function for the relative size of the basin of stability of the consensus equilibrium which is close to $\frac{1}{2}n\log n$, and 3) the impossibility of $n$-cycles in this system for $n>2$. These results represent first steps towards understanding coordination under myopic best response in this critical case, which in turn will improve our understanding of many processes from the selection of conventions among coworkers in an office to spatial distribution of language and culture, and how these may change with increasing interconnectedness. 

\section{Supplemental Materials} Accompanying this work is a supplemental file containing a visual catalogue of the partitions described in section \ref{n6}.

\section{Acknowledgments}
    Much of the computation that was crucial for this project was performed with the National Institute for Modeling Biological Systems (NIMBioS) computational resources at the University of Tennessee. Special thank you to Graham Derryberry for help using the cluster computer.

    Further, we would like to thank the members of the Fefferman Lab for commenting on this work. 

\section{Statements and Declarations}
    \subsection{Declaration of Competing interests} The Authors have no competing interests to declare.

    \subsection{Code Availability}\label{code}
    The code that was used to run the simulations and produce the catalogue can be found at \url{https://zenodo.org/doi/10.5281/zenodo.11237078}

    \subsection{Ethical Statements}
    This declaration is not applicable to this work.

    \subsection{Funding}
    This work was not funded. There are no funding sources to declare. 

\bibliography{StructuredCoordinationMainTextJune14.bib}

\end{document}